\documentclass[12pt]{article}

\usepackage[T1]{fontenc}
\usepackage{bbm}

\usepackage{amsfonts}
\usepackage{graphicx}
\usepackage{verbatim}
\usepackage{enumerate}
\usepackage{amsthm}
\usepackage[english]{cleveref}
\usepackage{amsmath}

\usepackage{fullpage}
\usepackage[numbers]{natbib}

\usepackage{pxfonts}
\allowdisplaybreaks

\newtheorem{theorem}{Theorem}[section]
\numberwithin{theorem}{section}
\newtheorem{corollary}[theorem]{Corollary}
\newtheorem{lemma}[theorem]{Lemma}

\newtheorem{definition}[theorem]{Definition}

\theoremstyle{definition}

\numberwithin{equation}{section}

\newcommand{\G}{\mathbb{G}}

\newcommand{\R}{\mathbb{R}}

\newcommand{\EEE}{\mathcal{E}}
\newcommand{\WW}{{\bf W}}
\newcommand{\GG}{{\bf G}}

\newcommand{\I}{\mathbb{I}}
\newcommand{\Ss} {{\mathcal S}}
\newcommand{\PPP} {{\mathbb P}}
\newcommand{\RRR} {{\mathcal R}}

\newcommand{\WWW}{\mathcal{W}}
\newcommand{\QQ} {{\mathcal Q}}

\newcommand{\UU}{{\bf U}}
\newcommand{\VV}{{\bf V}}
\newcommand{\FFF}{{\bf F}}

\newcommand{\ZZ}{{\bf Z}}
\newcommand{\JJ}{{\bf J}}
\newcommand{\HHH}{{\bf H}}

\def\GGG{\mathcal G}

\def\hhh{\mathrm{h}}
\def\reg{\mathrm{reg}}
\def\cut{\mathrm{cut}}
\def\tw{\mathrm{tv}}

\def \P {{\mathcal P}}

\def \_reg {\rightarrow_{\bf reg}}

\def\maxdeg/{\Delta}

\allowdisplaybreaks

\newcommand{\du}{\mathrm{d}}

\begin{document}

\title{\textbf{Explicit Bounds for Nondeterministically Testable Hypergraph Parameters}}
\author{Marek Karpinski\thanks{Dept. of Computer Science and the Hausdroff Center for Mathematics, University of Bonn. Supported in part by DFG grants, the Hausdorff grant EXC59-1/2. E-mail: \textrm{marek@cs.uni-bonn.de}}
	\and 
	Roland Mark\'o\thanks{Hausdorff Center for Mathematics, University of Bonn. Supported in part by a Hausdorff scholarship. E-mail: \textrm{roland.marko@hcm.uni-bonn.de}}}

\date{}
\maketitle
\begin{abstract}
	In this note we give a new effective proof method for the equivalence of the notions of testability and nondeterministic testability for uniform hypergraph parameters. We provide the first effective upper bound on the sample complexity of any nondeterministically testable $r$-uniform hypergraph parameter as a function of the sample complexity of its witness parameter for arbitrary $r$. The dependence is of the form of an exponential tower function with the height linear in $r$. Our argument depends crucially on the new upper bounds for the $r$-cut norm of sampled $r$-uniform hypergraphs. We employ also our approach for some other restricted classes of hypergraph parameters, and present some applications.
	{\normalsize }
\end{abstract}

\section{Introduction}

 The topic of property testing for combinatorial structures has gained considerable attention in recent years. In this setting in the case of graphs the goal is for a given property that is invariant under relabeling nodes to separate via sampling the set of graphs that satisfy it from those that are far from having the property. The development in this direction resulted in a number of randomized sub-linear time algorithms for the corresponding decision problems, see \cite{AFKSz}, \cite{AVKK2}, \cite{GGR}, for the background in approximation theory of NP-hard problems for dense structures, see \cite{AKK2}.
 
 Several attempts were made to characterize the properties in terms of the sample size needed for carrying out the above task, an important class comprises of those that admit a sample size that is independent from the size of the input instance, we call these properties testable, in other works they are also referred to as strongly testable. One particular family was introduced by \citet{LV} that consists of properties whose testability can be certified by some certain edge colorings, they called these nondeterministically (ND in short) testable properties, and are the main subject of the current note. It was showed by the authors of \cite{LV} that ND-testability is equivalent to testability for graphs, the question regarding parameters instead of properties was also discussed. Subsequently, a constructive proof was given for the above equivalence by \citet{GS}.
 The first treatment of parameters and properties of $r$-uniform hypergraphs ($r$-graphs in short) of higher order was carried out in \cite{KM3}, in that paper a proof was given for our main result \Cref{mainthm} below that relied in part on non-effective methods by means of the machinery developed in \cite{ESz} to describe the limit behavior of sequences of uniform hypergraphs.  

 The current note is based on the framework and terminology of \cite{KM3} by the authors, we will repeatedly refer to certain parts of \cite{KM3} for details, but also focus on delivering an accurate picture  by presenting the main steps here. 
 
 We proceed by providing the necessary formal definitions of the parameter testability in the dense hypergraph model. 
 
 \begin{definition}\label{ch6:def.partest}
 	An $r$-graph parameter $f$ is testable if for any $\varepsilon>0$ there exists a positive integer $q_f(\varepsilon)$ such that for any $q\geq q_f(\varepsilon)$ and simple $r$-graph $G$ with at least $q$ nodes we have
 	\begin{align*}
 	\PPP(|f(G)-f(\G(q,G))| >\varepsilon) < \varepsilon,
 	\end{align*}
 	where $\G(q,G)$ denotes the induced subgraph on a subset $S\subset V(G)$ of size $q$ chosen uniformly at random.
 	The infimum of the functions $q_f$ satisfying the above inequality is called the sample complexity of $f$.
 	The testability of parameters of $k$-colored $r$-graphs is defined analogously.
 \end{definition}
 
 One may relax the conditions of \Cref{ch6:def.partest} to introduce a certain version of nondeterministic testability. The definition below was first formulated in \cite{LV}.
% 
% An a priori weaker characteristic than the one above, nondeterministic testability, is the second cornerstone of the current work, and was introduced in \cite{LV}.
% 

 \begin{definition}\label{ch6:def.ndpartest}
 	An $r$-graph parameter $f$ is non-deterministically testable (ND-testable) if there exist an integer $k$ and a testable $2k$-colored $r$-graph parameter $g$ called witness such that for any simple graph $G$ the value $f(G)=\max_{\GG} g(\GG)$ where the maximum goes over the set of $k$-colorings of $G$ (see \Cref{sec.def} for the definition of a $k$-coloring). 
 	%The edge-$2k$-colored $r$-graph $\GG$ is a $k$-coloring of $G$, if we erase all edges of $\GG$ colored with an element of $[k+1, \dots, 2k]$ and forget about the coloring of the remaining edges, then we end up with $G$. %($G$ is the shadow of $\GG$).
 \end{definition}
 
 Our main result is the following. It was proved the first time in \cite{KM3} without any upper bound for the function $q_f$ in general, prior to that the case of $r=2$ was resolved by non-effective (\cite{LV}), and effective (\cite{GS},\cite{KM2}) methods.
 
 \begin{theorem}\label{mainthm}
 	Every non-deterministically testable $r$-graph parameter $f$ is testable. If $g$ is the parameter of $k$-edge-colored $r$-graphs that certifies the testability of $f$, then we have $q_f(\epsilon) \leq \exp^{(4(r-1)+1)}(c_{r,k} q_g(\epsilon)/\epsilon)$ for some constant $c_{r,k}>0$ depending only on $r$ and $k$, but not on $f$ or $g$. Here $\exp^{(t)}$ denotes the $t$-fold iteration of the exponential function for $t\geq1$, and $\exp^{(0)}$ is the identity function.
 \end{theorem}
 
 This note is organized as follows. We proceed by providing the necessarily notation in \Cref{sec.def}, in the subsequent \Cref{sec.efflemma} we state and prove the main technical development that enables us to discard the non-effective tools featured in the proof of the main result of \cite{KM3}. In \Cref{sec.main} we outline the  argument of \cite{KM3} that is slightly modified in order to adapt the framework to the concepts of the previous section, and present the proof of the main result \Cref{mainthm}. We illustrate a special case of ND-testable parameters, where an improvement on the sample complexity dependence is possible in \Cref{sec.linear}. In the subsequent \Cref{sec.appl} some applications of the main result are shown,  that are followed by a discussion together with the description of directions of possible further research in \Cref{sec.fur}.

\section{Preliminaries}\label{sec.def}

%The notion of an $r$-uniform hypergraph (in short $r$-graph in the current work) parameter and its testability can be defined completely analogously to the graph case, the same applies for non-deterministic testability, see \cite{LV},\cite{KM2}. Naturally the question arises whether or not the deterministic and non-deterministic testability are equivalent  for higher order hypergraphs, and if yes, then what can be said about the relationship of the sample complexity of the parameter and that of its witness parameter. We can settle the first problem by showing that the equivalence of the two notions persists also in this case. The second problem still remains open, more precisely a quantitative version of the result that the $r$-cut norm is preserved under sampling is required that provides a definite upper bound on the sample complexity in terms of the error parameter in the sense of \cite{KM2} replacing the weaker existential statement of the corresponding lemma.  

\subsection*{Hypergraphs and colored hypergraphs}

Simple $r$-uniform hypergraphs, $r$-graphs in short,  on $n$ vertices  forming the family $\GGG^r_n$ are subsets $G$ of ${[n] \choose r}$, the size of such a $G$ is $n$, and the elements of ${[n] \choose r}$ are $r$-edges. 
%Let $A_G$ denote the symmetric $\{0,1\}$-valued $r$-array or symmetric subset of $[n]^r$ that represents $G$. We will sometimes also use only the term $G$ to refer to a symmetric subset of $[n]^r \setminus \diag([n]^r)$ corresponding to the array representation, where $\diag([n]^r)$ is the set of vectors with at least one repetition in their coordinates. 
Let $k$ be a positive integer, and let $\GGG_n^{r,k}$ denote the set of $k$-colored $r$-graphs of size $n$, that are partitions $\GG=(G^\alpha)_{\alpha \in [k]}$ of ${[n] \choose r}$ into $k$ classes, we say that color $\alpha$ assigned to  $e$  ($\GG(e)=\alpha$) whenever $e\in G^\alpha$. In this sense simple $r$-graphs are regarded as $2$-colored. 
%In the $k$-colored case it is also possible to speak about the array representation, $A_{G^\alpha}$ stands for the symmetric $\{0,1\}$-valued $r$-array that represents the color class of $\alpha$, again with slight abuse of notation we will use $G^\alpha$ for $A_{G^\alpha}$. 
Additionally we have to introduce the special color $\iota$ for loop edges that are multisets of $[n]$ with cardinality $r$ having at least one element that has a multiplicity at least $2$. For any finite set $C$ the term $C$-colored graph is defined analogously to the $k$-colored case.

A {\it $k$-coloring} of a  $t$-colored $r$-graph $\GG=(G^\alpha)_{\alpha \in [t]}$ is a $tk$-colored $r$-graph $\hat \GG=(G^{(\alpha,\beta)})_{\alpha \in [t], \beta\in [k] }$ with colors from the set $[t] \times [k]$, where each of the original color classes indexed by $\alpha \in [t]$ can be retrieved by taking the union of the new classes corresponding to $(\alpha,\beta)$ over all $\beta \in [k]$, that is $G^\alpha=\cup_{\beta \in [k]}G^{(\alpha,\beta)}$. This last operation is called {\it $k$-discoloring} of a $[t] \times [k]$-colored graph, we denote it by $[\hat \GG,k]=\GG$.
We will sometimes write $tk$-colored for $[t] \times [k]$-colored graphs when it is clear from the context what we mean. 

Let $q \geq 1$ and $\GG \in \GGG_n^{r,k}$, then  $\G(q,\GG)$ denotes
the  random $r$-graph on $q$ vertices that is obtained by picking a subset $S$ of $[n]$ of cardinality $q$ uniformly at random and taking the induced subgraph $\GG[S]$. For any $\FFF \in \GGG_q^{r,k}$ and $\GG \in \GGG^{r,k}$ the $\FFF$-density of $\GG$ is defined as $t(\FFF,\GG)=\PPP(\FFF = \G(q, \GG))$.

\subsection*{Graphons}

Next we provide the description of the continuous generalization of $r$-graphs. We require some basic notation from the dense graph limit theory, and refer to \citet{L} for an extensive overview of recent developments in this topic.

For  a finite set $S$, let $\hhh(S)$ denote the set of nonempty subsets of $S$, %$\hhh([r-1],r-2)(S)$ the set of proper nonempty subsets of $S$,
and $\hhh(S,m)$ the set of nonempty subsets of $S$ of cardinality at most $m$. A $2^r-1$-dimensional real vector $x_{\hhh(S)}$ denotes $(x_{T_1}, \dots, x_{{T_{2^r-1}}})$, where $T_1, \dots, T_{2^r-1}$ is a fixed ordering of the nonempty subsets of $S$ with $T_{2^r-1}=S$, for a permutation $\pi$ of the elements of $S$ the vector  $x_{\pi(\hhh(S))}$ means $(x_{\pi'(T_1)}, \dots, x_{\pi'({T_{2^r-1}})})$, where $\pi'$ is the action of $\pi$ permuting the subsets of $S$. 

Let the $r$-kernel space $\mathcal{W}_0^r$ denote the space of the bounded measurable functions of the form $W\colon [0,1]^{\hhh([r],r-1)} \to \R$, and the subspace $\mathcal{W}^r$ of $\mathcal{W}_0^r$ the symmetric $r$-kernels that are invariant under coordinate permutations induced by $\pi \in S_r$, that is $W(x_{\hhh([r],r-1)})=W(x_{\pi(\hhh([r],r-1))})$ for each $\pi \in S_r$. We will refer to this invariance in the paper both for $r$-kernels and for measurable subsets of $[0,1]^{\hhh([r])}$ as {\it $r$-symmetric}.
Assume that the functions $W\in \mathcal{W}^r_{I}$ take their values in the interval $I$, for $I=[0,1]$ we call these special symmetric $r$-kernels  {\it{$r$-graphons}.} In what follows, $\lambda$ always denotes the usual Lebesgue measure in $\R^d$, where $d$ is everywhere clear from the context.

Analogously to the graph case we define the space of {\it $k$-colored $r$-graphons} $\mathcal W^{r,k}$ whose elements are referred to as $\WW=(W^\alpha)_{\alpha \in [k]}$ with each of the $W^\alpha$ components being an $r$-graphon. The special color $\iota$ that stands for the absence of colors has to be also employed in this setting as rectangles on the diagonal correspond to loop edges, see below for the case when we represent a $k$-colored $r$-graph as a graphon. The corresponding $r$-graphon $W^\iota$ is $\{0,1\}$-valued. Furthermore, $\WW$ has to satisfy $\sum_{\alpha \in [k]}W^\alpha(x)=1$ everywhere on $[0,1]^{\hhh([r],r-1)}$. For  $x \in [0,1]^{\hhh([r])}$ the expression $\WW(x)$ denotes the color at $x$, we have $\WW(x)=\alpha$ whenever  $\sum_{i=1}^{\alpha-1} W^i(x_{\hhh([r],r-1)}) \leq x_{[r]} \leq \sum_{i=1}^{\alpha} W^i(x_{\hhh([r],r-1)})$. 

%Let $\WW_\GG$ stand for the $k$-colored $r$-graphon that represents $\GG$. 
Similar to the  discrete case, a {\it $k$-coloring} of $\WW \in \mathcal W^{r,k}$ is a $tk$-colored $r$-graphon $\hat \WW=(W^{(\alpha,\beta)})_{\alpha \in [t], \beta\in [k] }$ with colors from the set $[t] \times [k]$ so that $\sum_{\alpha \in [t], \beta \in [k]} W^{(\alpha, \beta)}(x)=W^\alpha(x)$ for each $x \in [0,1]^{\hhh([r],r-1)}$ and $\alpha \in [t]$. The $k$-discoloring $[\hat \WW,k]$ of $\hat \WW$ and the term $C$-colored graphon is defined analogously, and simple $r$-graphons are treated as $2$-colored. 

For $q \geq 1$ and $\WW \in \mathcal W^{r,k}$ the random $[k]$-colored $r$-graph $\G(q, \WW)$ is generated as follows. The vertex set of $\G(q, \WW)$ is $[q]$, first we have to pick uniformly a random point $(X_S)_{S \in \hhh([q], r-1)} \in [0,1]^{\hhh([q], r-1)}$, then conditioned on this choice we conduct independent trials to determine the color of each edge $e \in {[q] \choose r}$ with the distribution given by $\PPP_e(\G(q, \WW)(e) = \alpha) = W^\alpha(X_{\hhh(e,r-1)})$ corresponding to $e$. Recall that $\iota$ is a special color which we want to avoid in most cases during the sampling process, therefore we will highlight the conditions that have to be imposed on the above random variables so that $\G(q, \WW) \in \GGG^{r,k}$.

For $\FFF \in \GGG_q^{r,k}$, the $\FFF$-density of $\WW$ is defined as $t(\FFF,\WW)=\PPP(\FFF = \G(q, \WW))$, which can be written following the above definition of the sampled random graph as 
\begin{align*}
	t(\FFF,\WW)= \int_{ [0,1]^{\hhh([q], r-1)} } \prod_{e \in {[q] \choose r}} W^{\FFF(e)}(x_{\hhh(e,r-1)}) \du \lambda(x_{\hhh([q], r-1)}).
\end{align*}

%The above notions were introduced in order to provide a concise representation for the limit space of $r$-graphs in \cite{ESz} and \cite{LSzlim}, in the current work we will not draw on this development explicitly but mention their relevance here. In a nutshell, a sequence of $r$-graphs converges if the corresponding numerical $\FFF$-density sequences converge for all $r$-graphs $\FFF$. One of the main results of \cite{LSzlim} for graphs and \cite{ESz} in the general case is that for every convergent sequence of $r$-graphs there exists an $r$-graphon they converge to in the sense that the $\FFF$-densities approach the $\FFF$-density of the limiting $r$-graphon. This was later reproved by \cite{Z} for general $r$ with purely combinatorial methods that are similar to concepts employed in the current paper. 

We can associate to each $\GG \in \GGG_n^{r, k}$ an element $\WW_\GG \in \mathcal W^{r,[k]\cup\{\iota\}}$ by subdividing the unit cube $[0,1]^{\hhh([r],1)}$ into $n^r$ small cubes the natural way and defining the function $W' : [0,1]^{\hhh([r],1)} \to [k]$ that takes the value $\GG(\{i_1, \dots, i_r\})$ on $[\frac{i_1-1}{n}, \frac{i_1}{n}] \times \dots \times [\frac{i_r-1}{n}, \frac{i_r}{n}]$ for distinct $i_1, \dots, i_r$, and the value $\iota$ on the remaining diagonal cubes, we will call such functions {\it naive $r$-graphons}. Then set $(\WW_\GG)^\alpha(x_{\hhh([r],r-1)}) = \I(W'(p_{\hhh([r],1)}(x_{\hhh([r],r-1)}))=\alpha)$ for each $\alpha \in [k] \cup\{\iota\}$, where $p_{\hhh([r],1)}$ is the projection to the suitable coordinates.
The special color $\iota$ here stands for the absence of colors has to be employed in this setting as rectangles on the diagonal correspond to loop edges. The corresponding $r$-graphon $W^\iota$ is $\{0,1\}$-valued. Note that 
\begin{align}\label{ch6:eq22}
	|t(\FFF,\GG)-t(\FFF,\WW_\GG)|\leq \frac{{q \choose 2}}{n - {q \choose 2}}
\end{align}
for each $\FFF \in \GGG_q^{r,k}$, hence the representation as naive graphons is compatible in the sense that $\lim_{n \to \infty} t(\FFF,\GG_n)=\lim_{n \to \infty} t(\FFF,\WW_{\GG_n})$ for any sequence $\{\GG_n\}_{n=1}^\infty$ with $|V(\GG_n)|$ tending to infinity.

%Our proof follows the proof skeleton of \cite{KM2}, but requires a more sophisticated approach.The reason for this is that the  analogous norm  for hypergraphs to the cut-norm that comes with a counting lemma has some shortcomings, for instance the sample is in some cases far away from the original in the natural distance notion induced by the norm. Therefore the corresponding regularity lemma cannot be applied directly as in \cite{KM2}. 

\subsection*{Norms and distances}

The definitions of the relevant norms is given next.

\begin{definition} \label{ch6:def.cutnorm}
	
%	Let $r \geq 1$ and $A$ be a real $r$-array of size $n$. Then the cut norm of $A$ is
%	\begin{align*}
%		\|A\|_{\square,r}=\frac{1}{n^r} \max_{\substack{S_i \subset [n]^{r-1} \setminus \diag ([n]^{r-1}) \\ i \in [r]}}
%		|A(r;S_1, \dots, S_r)|,
%	\end{align*}
%	where $A(r;S_1, \dots, S_r)=\sum_{i_1, \dots, i_r =1}^n A(i_1, \dots, i_r) \prod_{j=1}^r \I_{S_j} (i_1, \dots, i_{j-1}, i_{j+1}, \dots, i_r)$, and the maximum goes over sets $S_i$ that are invariant under coordinate permutations. This last requirement has to hold also below. 
%	
%	If $\P=(P_i)_{i=1}^t$ is a partition of $[n]^{r-1}\setminus \diag ([n]^{r-1})$ with symmetric classes, then the cut-$\P$-norm of $A$ is
%	\begin{align*}
%		\|A\|_{\square,r,\P}=\frac{1}{n^r} \max_{\substack{S_i \subset [n]^{r-1}}, i \in [r]}
%		\sum_{j_1, \dots, j_r=1}^t |A(r;S_1\cap P_{j_1} , \dots, S_r \cap P_{j_r})|.
%	\end{align*}
%	
	The cut norm of an $r$-kernel $W$ is
	\begin{align*}
		\|W\|_{\square,r}=\sup_{\substack{S_{i} \subset [0,1]^{\hhh([r-1])} \\ i \in [r]}}
		\left|\int_{ \cap_{i \in [r]} p_{[r]\setminus \{i\}}^{-1}(S_i)}W(x_{\hhh([r],r-1)}) \du \lambda(x_{\hhh([r],r-1)}) \right|,
	\end{align*}
	where the supremum is taken over $(r-1)$-symmetric measurable sets $S_i$, and $p_e$ is the natural projection from $[0,1]^{\hhh([r],r-1)}$ onto $[0,1]^{\hhh(e)}$. 
	Furthermore, for an $(r-1)$-symmetric partition $\P=(P_i)_{i=1}^t$ of $[0,1]^{\hhh([r-1])}$ the cut-$\P$-norm of an $r$-kernel is defined by the  formula
	\begin{align*}
		\|W\|_{\square,r,\P}=\sup_{\substack{S_{i} \subset [0,1]^{\hhh([r-1])} \\ i \in [r]}}
		\sum_{j_1, \dots, j_r=1}^t \left|\int_{\cap_{i \in [r]} p_{[r]\setminus \{i\}}^{-1}(S_i \cap P_{j_i} ) } W(x_{\hhh([r],r-1)}) \du \lambda(x_{\hhh([r],r-1)}) \right|,
	\end{align*}
	where the supremum is taken over sets $S_i$ that satisfy the usual symmetries.
\end{definition}
We remark that with the above definition it is also true that
\begin{align*}
	\|W\|_{\square,r}=\sup_{f_1, \dots, f_r \colon \mathcal [0,1]^{\hhh([r-1])} \to [0,1] }
	\left|\int\limits_{[0,1]^{\hhh([r],r-1)}} \prod_{i=1}^r f_i(p_{[r]\setminus \{i\}}(x) ) W(x_{\hhh([r],r-1)}) \du \lambda(x_{\hhh([r],r-1)}) \right|,
\end{align*}
where the supremum goes over $(r-1)$-symmetric functions $f_i$, and similarly for any  $(r-1)$-symmetric partition $\P=(P_i)_{i=1}^t$ of $[0,1]^{\hhh([r-1])}$ we have 
\begin{align*}
	\|W\|_{\square,r,\P}=\sup_{f_1, \dots, f_r \colon \mathcal [0,1]^{\hhh([r-1])} \to [0,1]}
	\sum_{j_1, \dots, j_r=1}^t  \left|\int\limits_{[0,1]^{\hhh([r],r-1)}} \prod_{i=1}^r f_i(x_{\hhh([r]\setminus \{i\})})\I_{P_{j_i}} (x_{\hhh([r]\setminus \{i\})}) W(x_{\hhh([r],r-1)}) \du \lambda(x) \right|.
\end{align*}
A relaxed variant of the $r$-cut norm is 
\begin{align*}
\|W\|_{\boxplus,r}=\sup_{f_1, \dots, f_r\colon \mathcal [0,1]^{\hhh([r-1])}\to [-1,1]}
\left|\int\limits_{[0,1]^{\hhh([r],r-1)}} \prod_{i=1}^r f_i(x_{\hhh([r]\setminus \{i\})}) W(x_{\hhh([r],r-1)}) \du \lambda(x_{\hhh([r],r-1)}) \right|.
\end{align*}
It is straightforward that $2^{-r}\|W\|_{\boxplus,r} \leq \|W\|_{\square,r} \leq \|W\|_{\boxplus,r}$ for every $r$ and $r$-kernel $W$. Note that for $r=2$ we have $\|W\|_{\boxplus,2}=\|T_W\|_{\infty \to 1}$, where $T_W$ is the integral operator from $L^{\infty}([0,1])$ to $L^1([0,1])$ with the kernel $W$.

We also mention that in several previous papers, see e.g. \cite{AVKK2}, the cut norm for $r$-arrays denotes a term that is significantly different from the one in \Cref{ch6:def.cutnorm} and is not suitable for our present purposes.
The above norms give rise to a distance between $r$-graphons, and analogously for $r$-graphs.

\begin{definition}\label{defdist}
	For two $k$-colored $r$-graphons $\UU=(U^\alpha)_{\alpha \in [k]}$ and $\WW=(W^\alpha)_{\alpha \in [k]}$, their cut distance is defined as 
	\begin{equation*}
		d_{\square,r}(\UU,\WW) = \sum_{\alpha=1}^k \|U^\alpha-W^\alpha\|_{\square,r},
	\end{equation*}
	and their cut-$\P$-distance as
	\begin{equation*}
		d_{\square,r,\P}(\UU,\WW) = \sum_{\alpha=1}^k \|U^\alpha-W^\alpha\|_{\square,r,\P}.
	\end{equation*}
	For two $k$-colored $r$-graphs $\GG=(G^\alpha)_{\alpha \in [k]}$ and $\HHH=(H^\alpha)_{\alpha \in [k]}$ their corresponding distances are defined as 
	\begin{align*}
		d_{\square,r}(\GG,\HHH)=d_{\square,r}(\WW_\GG,\WW_\HHH),
	\end{align*}
	and
	\begin{align*}
		d_{\square,r,\P}(\GG,\HHH)=d_{\square,r,\P}(\WW_\GG,\WW_\HHH).
	\end{align*}
	Distances between an $r$-graph and an $r$-graphon, as well as  in the case of $r$-kernels, are defined analogously.
\end{definition}

%Note that the norms introduced above are in general smaller or equal than the $1$-norm of integrable functions, also $d_{\square,r}(\UU,\WW)  \leq
%d_{\square,r,\P}(\UU,\WW)$ hods for every pair. 
%Their relevance will be clearer in the context of the next counting lemma, we include the standard proof only for completeness' sake. 
%
%
%\begin{lemma}
%	\label{ch6:hypercount}
%	Let $\UU$ and $\WW$ be two $k$-colored $r$-graphons with $\|U\|_\infty, \|W\|_\infty \leq 1$. Then for every $\FFF \in \GGG_q^{r,k}$ it holds that
%	\begin{align*}
%		|t(\FFF,\WW)-t(\FFF,\UU)| \leq {q \choose r} d_{\square,r}(\UU,\WW). 
%	\end{align*}
%\end{lemma}
%\begin{proof}
%	Fix $q$ and $\FFF \in \GGG_q^{r,k}$. Then
%	\begin{align*}
%		&|t(\FFF,\WW)-t(\FFF,\UU)| = |\int\limits_{ [0,1]^{\hhh([q], r-1)} } \prod_{e \in {[q] \choose r}} W^{\FFF(e)}(x_{\hhh(e,r-1)}) -\prod_{e \in {[q] \choose r}} U^{\FFF(e)}(x_{\hhh(e,r-1)}) \du \lambda(x)| \\
%		& \quad \leq \sum_{ e \in {[q] \choose r}} | \int\limits_{ [0,1]^{\hhh([q], r-1)} }  [W^{\FFF(e)}(x_{\hhh(e,r-1)})-U^{\FFF(e)}(x_{\hhh(e,r-1)}) ] \\ & \qquad \qquad \qquad \qquad \prod_{f \in {[q] \choose r}, f \prec e} W^{\FFF(f)}(x_{\hhh(f,r-1)})\prod_{g \in {[q] \choose r}, e \prec g} U^{\FFF(g)}(x_{\hhh(g,r-1)})\du \lambda(x)| \\
%		& \quad \leq  \sum_{ e \in {[q] \choose r}} \|W^{\FFF(e)}-U^{\FFF(e)}\|_{\square,r} \leq {q \choose r} d_{\square,r}(\UU,\WW),
%	\end{align*}
%	where $\prec$ is an arbitrary total ordering of the elements of ${q \choose r}$.
%\end{proof}

A generalization of the notion of a step function in the case of $2$-graphons (see \cite{BCL}) to the situation where we deal with $r$-graphons is given next. For a partition $\P$ the number of its classes is denoted by  $t_\P$.

\begin{definition}
	We call an $k$-colored $r$-graphon $\WW$ with $r \geq l$ an $(r,l)$-step function if there exist  positive integers $t_l, t_{l+1}, \dots, t_r=k$, an $l$-symmetric partition $\P=(P_1, \dots, P_{t_l})$ of $[0,1]^{\hhh([l])}$, and real arrays $A^\alpha_s: [t_{s-1}]^{\hhh([s],s-1)} \to [0,1]$  with $\alpha \in [t_s]$ for $l+1 \leq s \leq r$ such that $\sum_{\alpha \in [t_s]} A^\alpha_s(i_{\hhh([s],s-1)})=1$ for any choice of $i_{\hhh([s],s-1)}$ and for $s \leq r$ so that $W^\alpha$ for $\alpha \in [k]$ is of the following form for each $\alpha \in [k]$.
	\begin{align*}
		& W^\alpha(x_{\hhh([r],r-1)}) =  \sum_{\substack{i_S=1 \\ S \subset [r], l \leq |S|}}^{t_{|S|}} A^\alpha_r(i_{\hhh([r],r-1)}) \\ & \qquad\prod_{S \in {[r] \choose l}} \I_{P_{i_S}} (x_{\hhh(S)}) \prod_{\substack{S \subset [r] \\ l+1 \leq |S| < r}} \I( \sum_{j=1}^{i_{S}-1}A^j_{|S|}(i_{\hhh(S, |S|-1)})\leq x_S \leq \sum_{j=1}^{i_{S}}A^j_{|S|}(i_{\hhh(S, |S|-1)})).
	\end{align*}
	We refer to the partition $\P$ as the steps of $W$.
\end{definition}
The most simple example is the $(r,r-1)$ step function that can be written as 
\begin{equation*}
	W^\alpha(x_{\hhh([r],r-1)}) =  \sum_{i_1, \dots , i_r=1}^{t_{r-1}} A^\alpha_r(i_1,\dots, i_r)\prod_{j =1}^r \I_{P_{i_j}} (x_{\hhh([r] \setminus \{j\})}),
\end{equation*}
where $\P$ is an $(r-1)$-symmetric partition.

We speak of naive step functions when the arrays in the above definition have $\{0,1\}$ entries. In this case a color of an edge is determined by the membership of it subtuples in the classes of $\P$. Basically, an $(r,l)$-step function can be regarded as an interpolation between discrete and continuous objects, since $W(y_{h([l])}, . )$ is discrete for any fixed $y_{h([l])} \in [0,1]^{h([l])}$. 

\subsection*{Auxiliary lemmas}

We require three technical lemmas from \cite{KM3}, we refer to Sections 3 and 4 in \cite{KM3} for the proofs. The first asserts that any $r$-graphon can be approximated in a certain sense by an $r$-graphon of bounded complexity in terms of the permitted error and is a variant of the Regularity Lemma.

\begin{lemma}\label{ch6:sregkhyper}
	%Let $n \geq 1$ fixed, and let each partition in the statement be such that it is refined by the canonical partition of $[0,1]$ into $n$ parts, and each function  be such that it is constant on the product sets of classes of the canonical partition.
	
	For every $r,t,k \geq 1$, $\varepsilon>0$, and $k$-colored  $r$-graphon $\WW$ there exists an $(r-1)$-symmetric partition $\P=(P_1, \dots, P_m)$ of $[0,1]^{\hhh([r-1])}$ into $ m \leq (2t)^{(rk+1)^{4k^2/\varepsilon^2}}=t_\reg(r,k,\varepsilon, t)$ parts and an $(r-1)$-symmetric $(r,r-1)$- step function $\VV \in \WWW^{r,k}$ with steps from $\P$, such that for any partition $\QQ$ of $[0,1]^{\hhh([r-1])}$ into at most $m t$ classes we have
	\begin{align*}
	d_{\square,r,\QQ}(\WW,\VV) \leq \varepsilon. 
	\end{align*} 
	%Furthermore it holds that
	%\begin{align}
	%d_{\square,r,\P}(\WW,\WW_\P) \leq \varepsilon. 
	%\end{align}
	%If $\WW=\WW_\GG$ for some $k$-colored $\GG$ with $|V(\GG)|=n$, then one can require that $\P$ is an $n$-partition of $[0,1]^{\hhh([r-1])}$. 
	%If we want the parts to have almost equal measure, then the upper bound on the number of classes is modified into $2^{(4k^2+2)^{64/\varepsilon^2}}$.
\end{lemma}

The second lemma gives a sufficient condition for the existence of a coloring of a given $r$-graphon that is close to a fixed colored $r$-graphon.

\begin{lemma}\label{ch6:hypercoloring}
	Let $\varepsilon>0$, $\UU$ be a $t$-colored $r$-graphon that is an $(r,r-1)$-step function with steps $\P=(P_1, \dots, P_m)$ and $\VV$ be a $t$-colored $r$-graphon with $d_{\square,r, \P}(U, V) \leq \varepsilon$. For any $k \geq 1$ and $\hat \UU$ a $[t]\times [k]$-colored $r$-graphon that is an $(r,r-1)$-step function with steps from $\P$ such that $[\hat \UU, k]=\UU$ there exists a $k$-coloring of $\VV$ denoted by $\hat \VV$ so that 
	\begin{equation*}
	d_{\square,r, \P}(\hat \UU, \hat \VV) \leq k \varepsilon.
	\end{equation*}
	%$\|\UU-\VV\|_{\square,r}= \sum_{\alpha, \beta=1}^k \|U^{(\alpha, \beta)}-V^{(\alpha, \beta)}\|_\square \leq k^2\varepsilon.$
\end{lemma}

 Let $d_\tw$ denote the total variation distance between probability measures on $\GGG_q^{r,k}$. Let $\mu(q, G)$ and $\mu(q, W)$ denote the probability measure of $\G(q, G)$ and $\G(q, W)$ respectively. From (\ref{ch6:eq22}) it follows for each $\GG \in \GGG_n^{r,k}$ that 	
\begin{align}\label{eq12}
d_\tw (\mu(q, \WW_\GG), \mu(q, \GG)) \leq \frac{k^{q^r} q^2}{n}.
\end{align}
The third statement provides a an upper bound on the total variation distance of the probability measures of random $r$-graphs regarding their cut distance.

\begin{lemma} \label{ch6:hypercountcor}
	If $\UU$ and $\WW$ are two $k$-colored $r$-graphons, then 
	\begin{align*}
	d_\tw (\mu(q, \WW), \mu(q, \UU)) \leq \frac{k^{q^r} q^r}{2r!} d_{\square,r}(\UU,\WW),
	\end{align*}
	and there exists a coupling in form of $\GG_1$ and $\GG_2$ of the random $r$-graphs $\G(q, \WW)$ and $\G(q, \UU)$ respectively, such that 
	\begin{align*}
	P(\GG_1 \neq \GG_2) \leq \frac{k^{q^r} q^r}{2r!} d_{\square,r}(\UU,\WW). 
	\end{align*}
\end{lemma}

\section{Effective upper bound for the $r$-cut norm of a sampled $r$-graph}\label{sec.efflemma}

We are going to establish upper and lower bounds for the $r$-cut norm of an $r$-kernel using certain subgraph densities. Let $W$ be an $r$-kernel, and $H \subset {[q] \choose r}$ be a simple $r$-graph on $q$ vertices, define 
\begin{align*}
t^*(H,W)=\int_{ [0,1]^{\hhh([q], r-1)} } \prod_{e \in H} W(x_{\hhh(e,r-1)}) \du \lambda(x),
\end{align*}
this expression is a variant of the subgraph densities discussed above in \Cref{sec.def}. Using the previously introduced terminology we can write
\begin{align*}
t^*(H,W)=\sum_{H \subset F \subset {q \choose r}} t(F,W).
\end{align*}

Let $K_r^2$ denote the simple $r$-graph that is the $2$-fold blow-up of the $r$-graph consisting of $r$ vertices and one edge. That is, $V(K_r^2)=\{v^0_1, \dots, v^0_r, v^1_1, \dots, v^1_r\}$ and $E(K_r^2)=\{\{v^{i_1}_1, \dots, v^{i_r}_r\}: i_1, \dots, i_r \in \{0,1 \} \}$, alternatively we may regard $K_r^2$ as a subset of ${[2r] \choose r}$. 

It was shown in \citet{BCL} for $r=2$ with tools from functional analysis that for any symmetric $2$-kernel $W$ with $\|W\|_\infty \leq 1$  we have that
\begin{align}\label{ch6:eq41}
\frac{1}{4}t^*(K_2^2,W) \leq \|W\|_{\square,2} \leq [t^*(K_2^2, W)]^{1/4},
\end{align}
where $[t^*(K_2^2, .)]^{1/4}$is called the trace norm or the Schatten norm of the integral operator $T_W$.
We remark that in the above case $K_2^2$ stands for the $4$-cycle. Furthermore, it is not hard to show that for any $r$ and $r$-kernel $W$ is holds that $t^*(K_r^2, W) \geq 0$. It holds that
\begin{align*}
t^*(K_r^2,W)&=\int_{ [0,1]^{\hhh(V(K_r^2), r-1)} } \prod_{i_1, \dots, i_r \in\{0,1\}} W(x_{\hhh(\{v_1^{i_1}, \dots, v_r^{i_r}\} ,r-1)}) \du \lambda(x) \\
&=\int_{ [0,1]^{T_1}} \int_{ [0,1]^{T_2} } \prod_{i_1, \dots, i_r \in\{0,1\}} W(x_{\hhh(\{v_1^{i_1}, \dots, v_r^{i_r}\} ,r-1)}) \du \lambda(x_{T_2}) \du \lambda(x_{T_1}) \\
&=\int_{ [0,1]^{T_1}} \left[\int_{ [0,1]^{T^0_2}} \prod_{i_1, \dots, i_{r-1} \in\{0,1\}} W(x_{\hhh(\{v_1^{i_1}, \dots, v_{r-1}^{i_{r-1}},v_r^0\} ,r-1)}) \du \lambda(x_{T^0_2 }) \right] \\ 
&\quad \left[\int_{ [0,1]^{T^1_2 }} \prod_{i_1, \dots, i_{r-1} \in\{0,1\}} W(x_{\hhh(\{v_1^{i_1}, \dots, v_{r-1}^{i_{r-1}},v_r^1\} ,r-1)}) \du \lambda(x_{T^1_2}) \right]
\du \lambda(x_{T_1}) \\
&= \int_{ [0,1]^{T_1}} \left[\int_{ [0,1]^{T_3\setminus T_1} } \prod_{i_1, \dots, i_{r-1} \in\{0,1\}} W(x_{\hhh(\{v_1^{i_1}, \dots, v_{r-1}^{i_{r-1}}, u\} ,r-1)}) \du \lambda(x_{T_3\setminus T_1})\right]^2 \du \lambda(x_{T_1}) ,
\end{align*}
where $T_1=\hhh(V(K_r^2)\setminus \{v_r^0,v_r^1\}, r-1)$, $T^i_2=\hhh(V(K_r^2), r-1) \setminus \hhh(V(K_r^2)\setminus \{v_r^{i+1}\}, r-1)$ is the subset of $\hhh(V(K_r^2), r-1)$ whose elements contain $v_r^i$, but not  $v_r^{i+1}$ for $i \in\{0,1\}$, $T_2=T_2^0 \cup T_2^1$, and $T_3=\hhh(V(K_r^2)\setminus \{v_r^0,v_r^1\} \cup \{u\}, r-1)$. We used Fubini's theorem, that enabled us to integrate first over coordinates with indices from $T_2$, which we could then use to identify $v_r^0$ and $v_r^1$. 

In the proof of (\ref{ch6:eq41}) the authors drew on tools from functional analysis and the fact that a $2$-kernel describes an integral operator, those concepts do not have a natural counterpart for $r$-kernels. However we can provide an analogous result by the repeated application of Fubini's theorem and the Cauchy-Schwarz inequality in the $L^2$-space.
\begin{lemma} \label{ch6:lemmacuteff} For any $r\geq1$ and $r$-kernel $W$ with $\|W\|_\infty \leq 1$ we have 
\begin{align}
2^{-r} t^*(K_r^2, W) \leq \|W\|_{\square, r} \leq [t^*(K_r^2, W)]^{1/2^r}.
\end{align}
\end{lemma}
\begin{proof}

The lower bound on $\|W\|_{\square, r}$ is straightforward, and $K_r^2$ could even be replaced by any other simple $r$-graph, we only need to use $2^{-r} \|W\|_{\boxplus, r} \leq \|W\|_{\square, r}$.

For the other direction, let us fix a collection of arbitrary symmetric measurable functions $f_1, \dots, f_r\colon \mathcal [0,1]^{\hhh([r-1])}\to [0,1]$. Set $V=\{v_1, \dots, v_r\}$ and for any $l \geq 1$ and $i_1,\dots, i_l \in \{0,1\}$ let $$V^{i_1, \dots, i_l}=\{v_1^{i_1}, \dots, v_l^{i_l}, v_{l+1}, \dots, v_r\}.$$ Further, let $V_j=V \setminus \{v_j\}$ and for $j \geq l+1$ let $V_j^{i_1, \dots, i_l}=V^{i_1, \dots, i_l} \setminus \{v_j\}$. Let us introduce the index sets $T_1=\hhh(V_1) $, $S_1=\hhh(V,r-1) \setminus T_1$ and for $1 \leq l\leq r$ the sets

$$S_l^0=\{(e \setminus \{v_l\}) \cup \{v^0_l\}| e \in S_l \}, \qquad S_l^1=\{(e \setminus \{v_l\}) \cup \{v^0_1\}| e \in S_l \},$$

$$T_{l+1}=\cup_{i_1, \dots,i_{l} \in \{0,1\}} \hhh(V^{i_1, \dots,i_{l}}_{l+1}),$$

and

$$S_{l+1}=(T_{l}\cup S^0_{l} \cup S^1_{l}) \setminus T_{l+1}.$$

Then we have

\begin{align*}
&\left|\int_{[0,1]^{\hhh(V,r-1)}} \prod_{j=1}^r f_j(x_{\hhh(V_j)}) W(x_{\hhh(V,r-1)}) \du \lambda(x_{\hhh(V,r-1)})\right| \\ 
& \quad =
\left| \int_{[0,1]^{T_1}} f_1(x_{\hhh(V_1)}) \int_{[0,1]^{S_1}} \prod_{j=2}^{r} f_j(x_{\hhh(V_j)}) W(x_{\hhh(V,r-1)}) \du \lambda(x_{S_1})\du \lambda(x_{T_1}) \right| \\ 
& \quad \leq
\left[ \int_{[0,1]^{T_1}} f^2_1(x_{\hhh(V_1)}) \lambda(x_{T_1}) \right]^{1/2} \left[ \int_{[0,1]^{T_1}} \left( \int_{[0,1]^{S_1}} \prod_{j=2}^{r} f_j(x_{\hhh(V_j)}) W(x_{\hhh(V,r-1)}) \du \lambda(x_{S_1})\right)^2 \du \lambda(x_{T_1}) \right]^{1/2}   \\ 
& \quad \leq  \Bigg[ \int_{[0,1]^{T_1}} \left( \int_{[0,1]^{S^0_1}} \prod_{j=2}^{r} f_j(x_{\hhh(V^0_j)}) W(x_{\hhh(V^0,r-1)}) \du \lambda(x_{S^0_1})\right) \\ & \qquad \qquad \qquad \qquad \left( \int_{[0,1]^{S^1_1}} \prod_{j=2}^{r} f_j(x_{\hhh(V^1_j)}) W(x_{\hhh(V^1,r-1)}) \du \lambda(x_{S^1_1})\right) \du \lambda(x_{T_1}) \Bigg]^{1/2},  
\end{align*}
where we used $\|f_1\|_\infty \leq 1$ and the identity $\int (\int f(x,y) \du y)^2 \du x= \int f(x,y)f(x,z) \du y\du z \du x$ in the previous inequality. We proceed by upper bounding the last expression through repeated application of this reformulation combined with Cauchy-Schwartz.
% \\
% & \qquad \qquad \qquad \vdots  \\
 \begin{align*}
  &   \Bigg[\int_{[0,1]^{T_l}} \left(\prod_{i_1, \dots,i_{l-1}\in \{0,1\}} f_l(x_{\hhh(V^{i_1, \dots,i_{l-1}}_l)}) \right) \\ & \qquad \qquad \qquad \qquad \left(\int_{[0,1]^{S_l}} \prod_{i_1, \dots,i_{l-1}\in \{0,1\}} \prod_{j=l+1}^{r} f_j(x_{\hhh(V^{i_1, \dots,i_{l-1}}_j)}) W(x_{\hhh(V^{i_1, \dots,i_{l-1}})}) \du \lambda(x_{S_l}) \right) \du \lambda(x_{T_l}) \Bigg]^{\frac{1}{2^{l-1}}} \\
 & \quad \leq \left[\int_{[0,1]^{T_l}} \left(\prod_{i_1, \dots,i_{l-1}\in \{0,1\}} f_l(x_{\hhh(V^{i_1, \dots,i_{l-1}}_l)}) \right)^2 \lambda(x_{T_l}) \right]^{\frac{1}{2l}} \\ & \qquad \qquad \qquad \qquad \left[\int_{[0,1]^{T_l}}  \left(\int_{[0,1]^{S_l}} \prod_{i_1, \dots,i_{l-1}\in \{0,1\}} \prod_{j=l+1}^{r} f_j(x_{\hhh(V^{i_1, \dots,i_{l-1}}_j)}) W(x_{\hhh(V^{i_1, \dots,i_{l-1}})}) \du \lambda(x_{S_l}) \right)^2 \du \lambda(x_{T_l})\right]^{\frac{1}{2^l}} \\
 & \quad \leq \Biggl[\int_{[0,1]^{T_l}}  \left(\int_{[0,1]^{S^0_l}} \prod_{i_1, \dots,i_{l-1}\in \{0,1\}} \prod_{j=l+1}^{r} f_j(x_{\hhh(V^{i_1, \dots,i_{l-1},0}_j)}) W(x_{\hhh(V^{i_1, \dots,i_{l-1},0})}) \du \lambda(x_{S^0_l}) \right) \\ & \qquad \qquad \qquad \qquad \left(\int_{[0,1]^{S^1_l}} \prod_{i_1, \dots,i_{l-1}\in \{0,1\}} \prod_{j=l+1}^{r} f_j(x_{\hhh(V^{i_1, \dots,i_{l-1},1}_j)}) W(x_{\hhh(V^{i_1, \dots,i_{l-1},1})}) \du \lambda(x_{S^1_l}) \right) \du \lambda(x_{T_l})\Biggr]^{\frac{1}{2^l}} \\
 &\quad =  \left[\int_{[0,1]^{T_l \cup S_l^0 \cup S_l^1}}   \prod_{i_1, \dots,i_{l}\in \{0,1\}} \prod_{j=l+1}^{r} f_j(x_{\hhh(V^{i_1, \dots,i_{l}}_j)}) W(x_{\hhh(V^{i_1, \dots,i_{l}})}) \du \lambda(x_{S^0_l}) \du \lambda(x_{S^1_l}) \du \lambda(x_{T_l})\right]^{\frac{1}{2^l}} \\
 &\quad =  \Bigg[\int_{[0,1]^{T_{l+1}}} \left(\prod_{i_1, \dots,i_{l}\in \{0,1\}} f_{l+1}(x_{\hhh(V^{i_1, \dots,i_{l}}_{l+1})}) \right) 
 \\ & \qquad \qquad \qquad \qquad
 \left(\int_{[0,1]^{S_{l+1}}} \prod_{i_1, \dots,i_{l}\in \{0,1\}} \prod_{j=l+2}^{r} f_j(x_{\hhh(V^{i_1, \dots,i_{l}}_j)}) W(x_{\hhh(V^{i_1, \dots,i_{l}})}) \du \lambda(x_{S_{l+1}}) \right) \du \lambda(x_{T_{l+1}})\Bigg]^{\frac{1}{2^l}}  \\ & \qquad \qquad \qquad \vdots \\
 &\quad \leq  \left[\int_{[0,1]^{S_{r}}} \prod_{i_1, \dots,i_{r}\in \{0,1\}} W(x_{\hhh(V^{i_1, \dots,i_{r}})}) \du \lambda(x_{S_{r}})\right]^{\frac{1}{2r}}=t^*(K_r^2,W)^{1/2^r},
\end{align*}
where in subsequent inequalities we first used the Cauchy-Schwarz inequality, and afterwards that $\|f_j\|_\infty \leq 1$ for any $j \in [r]$. As the test functions $f_1, \dots, f_r$ were arbitrary the statement of the lemma follows.
\end{proof}

Utilizing the previous result we can obtain a quantitative upper bound on the cut norm of the sampled kernel for arbitrary $r$.

\begin{lemma}\label{ch6:hyperregpres2}
Let $r,k \geq 1$. For any $\varepsilon >0 $ and $t \geq 1$ there exists an integer $q_{\cut}(r,k,\varepsilon,t)\leq c (1/\epsilon)^{2^{2r}} t^{2^{2r}} k^3 r^2$ for some  universal constant $c>0$ such that  for any  $k$-tuple of $r$-kernels $U_1, \dots, U_k$ that take values in $[-1, 1]$, and any integer $q \geq q_{\cut}(r,k,\varepsilon,t)$ it holds with probability at least $1- \varepsilon$ that if
$$
  \sum_{l=1}^{k}\|U_l\|_{\square,r} \leq \left(\frac{\varepsilon}{k t^r}\right)^{2^r}2^{-r-1},$$
 then
 $$ \sup_{\QQ, t_\QQ\leq t} \sum_{l=1}^{k} \|W_{\G(q,U_l)}\|_{\square,r, \QQ} \leq \varepsilon.
$$ 
 where the supremum at both places goes over symmetric partitions $\QQ$ of $[0,1]^{\hhh([r-1])}$ into at most $t$ classes.
\end{lemma}
\begin{proof}
 Let $r,k,t \geq 1$ and $\varepsilon>0$ be fixed, and let $U_1, \dots, U_k$ and $q$ be arbitrary. It is a standard sampling result that for any $r$-kernel $U$, positive integer $q$, and $F \in \GGG^r$ we have that $$\PPP(|t^*(F, U)-t^*(F,\G(q,U))| \geq \delta) \leq 2 \exp(-\frac{\delta^2 q}{2|V(F)|^2})$$ for any $\delta>0$, in particular for $F=K_r^2$ we have 
$$\PPP(|t^*(K_r^2, U)-t^*(K_r^2,\G(q,U))| \geq \delta/2) \leq 2 \exp(-\frac{\delta^2 q}{32r^2}).$$

Then we can estimate $\sup_{\QQ, t_\QQ\leq t} \sum_{l=1}^{k} \|W_{\G(q,U_l)}\|_{\square,r, \QQ}$ using \Cref{ch6:lemmacuteff}. Set $\delta=\frac{1}{2k}\left(\frac{\varepsilon}{t^r}\right)^{2^r}$, and let $q$ be as large such that $2k\exp\left(-\frac{\delta^2 q}{32r^2}\right) < \varepsilon$. Let $\mathbb{A}$ denote the set of all $r$-arrays of size $t$ with $\{-1,1\}$ entries. Then we have

\begin{align*}\label{ch6:eq14}
&\sup_{\QQ, t_\QQ\leq t} \sum_{l=1}^{k} \|W_{\G(q,U_l)}\|_{\square,r, \QQ} \\ & =
\sup_{\QQ, t_\QQ\leq t}\max_{A \in \mathbb A}\sup_{\substack{T^l_j \subset [0,1]^{\hhh([r-1])}  \\ j \in [r], l \in [k]  }} \\ & \qquad
\sum_{l=1}^k \sum_{i_1, \dots, i_r=1}^t A(i_1, \dots, i_r) \int\limits_{[0,1]^{\hhh([r],r-1)}} W_{\G(q,U_l)}(x_{\hhh([r],r-1)}) \prod_{j=1}^r \I_{T^l_j \cap Q_{i_j}} (x_{\hhh([r] \setminus \{j\})})    \du \lambda(x_{\hhh([r],r-1)})\\
&\leq t^r \sum_{l=1}^k \|W_{\G(q,U_l)}\|_{\square,r} \\
&\leq t^r \sum_{l=1}^k t^*(K_r^2,W_{\G(q,U_l)})^{1/2^r} \\
&\leq t^r \sum_{l=1}^k (t^*(K_r^2,\G(q,U_l)) + \frac{4r^2}{q})^{1/2^r} \\
&\leq t^r \sum_{l=1}^k (t^*(K_r^2,U_l)+\delta)^{1/2^r} \\
&\leq t^r \sum_{l=1}^k ( 2^r\|U_l\|_{\square,r}+\delta)^{1/2^r}\leq \varepsilon,
\end{align*}
and the assumptions of the calculation, in particular the fourth inequality, hold true with probability at least $1-\varepsilon$. For convenience, the first inequality is true by definition, the third holds by (\ref{ch6:eq22}), whereas the second and the fifth are the consequence of \Cref{ch6:lemmacuteff}. 

%The lower bound on $\sup_{\QQ, t_\QQ\leq t} \sum_{j=1}^{k} \|W_{\G(q,U_j)}\|_{\square,r, \QQ}$ can be obtained by the same argument as above using Markov's inequality. For the upper bound we again discretize to obtain the $r$-kernels $W_1, \dots, W_k$ with common range $\{y_i: \alpha \in [m]\}$ such that $\|U_j-W_j\|_\infty \leq \frac{\varepsilon}{4k}$ for each $j \in [k]$, hence $m=\frac{8k}{\varepsilon}$ will do. We associate to each $W_j$ and $m$-colored $r$-graphon $\WW_j$ as above and set $J_A^\alpha$ to $y_\alpha(A \otimes B_0)$, then 
%\begin{align*}
%\max_{A_1, \dots, A_k \in \mathbb A} \sup_{\QQ, t_\QQ\leq t}\sum_{j=1}^k \EEE_{\QQ,r-1}(\WW_j, J_{A_j}) = \sup_{\QQ, t_\QQ\leq t} \sum_{j=1}^k\|W_j\|_{\square,r, \QQ}.
%\end{align*} 
%Similarly, 
%$$ \max_{A_1, \dots, A_k \in \mathbb A} \sup_{\QQ, t_\QQ\leq t} \sum_{j=1}^k \EEE_{\QQ,r-1}(\WW_{\G(q,\WW_j)}, J_{A_j})=\sup_{\QQ, t_\QQ\leq t} \sum_{j=1}^k\|W_{\G(q,W_j)}\|_{\square,r, \QQ}.$$
%The testability of $\sup_{\QQ, t_\QQ\leq t}\sum_{j=1}^k \EEE_{\QQ,r-1}(\WW_j, J_{A_j})$ follows from \Cref{ch5:samp} with a slight modification of the argument for any fixed tuple $A_1, \dots, A_k \in \mathbb A$. As the cardinality of $\mathbb A$ does not depend on $\WW_1, \dots, \WW_k$ the statement of the lemma follows.
\end{proof}

%\begin{align}
%&|\int_{T_0} \prod_{i=1}^r f_i(x_{T_0^i}) W(x_{T_0}) \du \lambda(x_{T_0})| =
%| \int_{T_1} \int_{S_1} \prod_{i=1}^{r-1} f_i(x_{T_0^i}) W(x_{T_0}) f_r(x_{T_0^r}) \du \lambda(x_{S_1})\du \lambda(x_{T_1})| \\ & \quad \leq 
% \left( \int_{T_1} [\int_{S_1} \prod_{i=1}^{r-1} f_i(x_{T_0^i}) W(x_{T_0}) \du \lambda(x_{S_1})]^2 \du \lambda(x_{T_1}) \right)^{1/2} \left( \int_{T_1} [\int_{S_1} f_r(x_{T_0^r})  \du \lambda(x_{S_1})]^2 \du \lambda(x_{T_1}) \right)^{1/2} \\ & \quad \leq  \left( \int_{T_1} \int_{K_1} \prod_{i_r=0}^1 \prod_{i=1}^{r-1} f_i(x_{T_0^i}) W(x_{T^{i_r}_0}) \du \lambda(x_{S_1}) \du \lambda(x_{T_1}) \right)^{1/2} 
%\end{align}

\section{Proof of the main result}\label{sec.main}
The next lemma is a crucial component in the proof of the main result.

\begin{lemma}\label{mainlemma}
	For every $r,t,k, q_0 \geq 1$ and $\delta > 0$	there exists an integer $q_\tw=q_\tw(r,\delta,q_0, t,k) \geq 1$ such that for every $q \geq q_\tw$ the following holds. Let $\UU=(U^\alpha)_{\alpha \in [t]}$ be a $t$-colored $r$-graphon and let $V^\alpha$ denote $W_{\G(q,U^\alpha)}$ for each $\alpha \in [t]$, also let $\VV=(V^\alpha)_{\alpha \in [t]}$, so  $\WW_{\G(q,\UU)}=\VV$. Then with probability at least $1-\delta$ there  exist for every $k$-coloring $\hat \VV=(V^{\alpha,\beta})_{\alpha \in [t], \beta \in [k]}$ of $\VV$ a $k$-coloring $\hat \UU=(U^{\alpha,\beta})_{\alpha \in [t], \beta \in [k]}$  of $\UU=(U^\alpha)_{\alpha \in [t]}$ %and a structure preserving permutation $\phi$ of $[0,1]^{r_{<}([r])}$  
	such that we have
	
	\begin{equation*}
	d_\tw (\mu(q_0, \hat \VV), \mu(q_0, \hat \UU)) \leq \delta.
	%\sum_{i=1}^t  \sum_{\FFF: |V(\FFF)|=q_0} |t(\FFF,\UU_i)-t(\FFF,\VV_i)| %\|(U^j_i)-V^j_i\|_{\square,r} 
	%\sum_{i=1}^t \sum_{j=1}^k \sup_{\QQ: t_{\QQ} \leq t_2}\|U^j_i-V^j_i\|_{\square,r,\QQ} \leq \delta
	\end{equation*}  
	The bound $q_\tw(r,\delta,q_0, t,k)$ can be chosen in a way so that $q_\tw(r,\delta,q_0, t,k)\leq \exp^{(4(r-1))} c_r (\frac{q_0^{r}}{\delta})^3 (kt)^{6q_0^r}$ for some constant $c_r>0$ only depending on the dimension $r$.
\end{lemma}
	The proof is to large extent identical to the proof of Lemma 5.1 in \citet{KM3}, the only part that is changed is where we replace the non-effective ultralimit method used in that proof by  \Cref{ch6:hyperregpres2}. However, the two statements that are exchanged do not coincide, thus some technical adjustment needs to be carried out. Next we present the sketch of the proof of \Cref{mainlemma} by outlining the main steps, for the details we refer to  Lemma 5.1 in \citet{KM3}. 
\begin{proof}

	We proceed by induction with respect to $r$. The case of $r=1$ can be verified the same way as in \cite{KM3}, and $q_\tw(1, \delta, q_0,t,k)= \frac{(t+\ln 2-\ln \delta) 3q_0^{2k+2}}{4\delta^2}$  satisfies the conditions of the lemma. 
	
	Now assume that we have already verified the statement of the lemma for $r-1$ and any other choice of the other parameters of $q_\tw$. We will conduct the proof for the case for $r$-graphons, therefore let   $\delta > 0$, $t,k,q_0 \geq 1$ be arbitrary and fixed, $q$ is to be determined below and let $\UU$, $\VV$, and $\hat \VV$ be as in the conditions of the lemma. We outline the steps in order to obtain a $k$-coloring $\hat \UU$ for $\UU$. 
	
	Let $\Delta= \Pi(r,\delta,q_0, t,k)=\frac{ \delta}{4k (kt)^{q_0^r} q_0^r}$. Set $t_2=t_\reg(r,tk,\Delta, 1)$ and $t_1=t_\reg(r,t,(\Delta/t_2^r t)^{2^r}2^{-r-1} , t_2)$, and define $q_\tw(r,\delta,q_0, t,k)=\max\{q_\tw(r-1,\delta/4,q_0,t_1,t_2), q_\cut(r,t,\Delta, t_2)\}$. Note that $t_2 \leq \exp^{(2)}(c (1/\Delta)^3)$ and $t_1 \leq \exp^{(4)}(c (1/\Delta)^3)$ for a large enough constant $c>0$. We also assume that $q_\tw(r-1,\delta,q_0, t,k) \leq \exp^{(d)} (c_{r-1} \left( \frac{1}{\Delta'}\right)^3)$ for some positive integer $d$ and real $c_{r-1}>0$, where $\Delta'=\Pi(r-1,\delta,q_0, t,k)=\frac{ \delta}{4k (kt)^{q_0^{r-1}} q_0^{r-1}}$. Then it follows 
	\begin{align}
    q_\tw(r-1,\delta/4,q_0,t_1,t_2)\leq \exp^{(d+4)}(c_r (1/\Delta)^3)
	\end{align}
	for some $c_r >0$. Since we can adjust the constant factor $c_{r-1}$ in a way that $q_\tw(r-1,\delta/4,q_0,t_1,t_2)\geq  q_\cut(r,t,\Delta, t_2)\}$ for any possible choice of the parameters we conclude that  $q_\tw(r,\delta,q_0, t,k)$ is upper bounded by $\exp^{(4(r-1))}(c_r (1/\Delta)^3).$
	 Let $q \geq q_\tw(r,\delta,q_0, t,k)$ be arbitrary.  We describe now the step for the construction of $\hat \UU$ that satisfies the conditions of the lemma.
	
	\begin{itemize}
		\item We approximate $\hat \VV$ by some function $\hat \ZZ$ that is only given implicitly by means of \Cref{ch6:sregkhyper}.  We have 
		\begin{equation*}
		\sup_{\QQ, t_\QQ \leq t_\RRR} d_{\square,r,\QQ}(\hat \VV, \hat \ZZ) \leq \Delta,
		\end{equation*} 
		where $\RRR$ denotes the set of the steps of $\hat \ZZ$, and $t_\RRR \leq t_2$ holds. 
		\item We set $\ZZ=[\hat \ZZ,k]$, consequently
		\begin{equation}\label{eq8}
		\sup_{\QQ, t_\QQ \leq t_\RRR} d_{\square,r,\QQ}( \VV,  \ZZ) \leq \Delta.
		\end{equation}
		Note that $\ZZ$ and $\hat \ZZ$ depend on $\hat \VV$. 
		\item We apply again \Cref{ch6:sregkhyper} with the proximity parameter $\Delta/2$ to $r$-graphon $\UU$ to approximate it by  $\WW_1=(W_1^1, \dots, W_1^t)$ with steps in $\P$ that satisfies 
		\begin{equation*}
		\sup_{\QQ, t_\QQ \leq t_\P t_2}d_{\square,r,\QQ}(\WW_1,\UU) \leq (\Delta/t_2^r t)^{2^r}2^{-r-1},
		\end{equation*} where the supremum runs over all $(r-1)$-symmetric partitions $\QQ$ of $[0,1]^{\hhh([r-1])}$ with at most $t_\P t_2$ classes, and $t_\P \leq t_1.$
		
		\item Define $\WW_2=(W_2^\alpha)_{\alpha \in [t]}$ to be the $r$-graphon representing $\G(q,\WW_1)$, so $W_2^\alpha$ represents $\G(q,W_1^\alpha)$ for each $\alpha \in [t]$. The steps of $\WW_2$ are denoted by $\P'$. Then it follows from \Cref{ch6:hyperregpres2} that 
		\begin{equation*}
		\sup_{\QQ, t_\QQ \leq  t_2}d_{\square,r,\QQ}(\WW_2,\VV) \leq \Delta, 
		\end{equation*}  
		with probability at least $1-\Delta$, so consequently
		\begin{equation*}
		d_{\square,r,\RRR}(\WW_2,\VV) \leq \Delta, 
		\end{equation*}  
		with the same failure probability. Furthermore, with (\ref{eq8}) we have
		\begin{equation} \label{eq10}
		d_{\square,r,\RRR}(\WW_2, \ZZ) \leq 2\Delta. 
		\end{equation} 
		\item We define the $k$-coloring $\hat \WW_2$ of $\WW_2$ via \Cref{ch6:hypercoloring}, which by  (\ref{eq10}) certifies the existence of a $k$-coloring such that 
		\begin{equation*}
		d_{\square,r,\RRR}(\hat \ZZ, \hat \WW_2) \leq 2k\Delta.
		\end{equation*}            
		The graphon $\hat \WW_2$ is a symmetric step function with steps that form the coarsest partition that refines both $\P'$ and $\RRR$, we denote this $(r-1)$-symmetric partition of $[0,1]^{\hhh([r-1])}$ by $\Ss$, the number of its classes satisfies $t_{\Ss} = t_{\P'} t_{\RRR} \leq t_1t_2$. 
		\item We construct the $k$-coloring $\hat \WW_1$ of $\WW_1$ using the hypothesis that  the current lemma is true for the case of $r-1$ and the arbitrary choice of all other parameters. For the details we refer to the proof in \cite{KM3}. The $r$-graphon $\hat \WW_1$ we obtained satisfies
		\begin{equation*}
	d_\tw (\mu(q_0,\hat \WW_1),\mu(q_0,\hat \WW_2 )) \leq \delta/4
		\end{equation*}
		with probability at least $1-\delta/4.$  Also, $\hat \WW_1$ has at most $t_\P t_2$ steps that refine $\P.$    
		
		\item  \Cref{ch6:hypercoloring} provides the existence of $\hat \UU$ with $[\hat \UU,k]=\UU$ with the bound 
		as $d_{\square,r}(\hat\UU,\hat \WW_1) \leq \frac{k\Delta}{2}.$
	\end{itemize}
	
	We conclude the proof by invoking \Cref{ch6:hypercountcor} to verify that $\hat \UU$ satisfies the conditions of the lemma,
	\begin{equation*}
		d_\tw (\mu(q_0, \hat \VV), \mu(q_0, \hat \UU)) \leq \delta,
	\end{equation*}   and the failure probability is at most $\delta.$

\end{proof}

\begin{proof}[Proof of \Cref{mainthm}]
We proceed completely identically to the proof of the main result of \cite{KM3}, we only have to substitute the current \Cref{mainlemma} for Lemma 5.1 in that paper, we only give a brief overview here, we refer for details to \cite{KM3}. Set $q_0=q_g(\varepsilon/4)$. The main observation is that provided the result of \Cref{mainlemma} we can find for any coloring $\FFF$ of $F$ a  coloring $\GG$ of $G$ such that the distributions of $\G(q_0,\WW_\GG)$ and $\G(q_0,\WW_\FFF)$ are close, hence the random objects given by them can be coupled in a way so that they coincide with high probability. Applying this together with the triangle inequality 
\begin{align*}
|g(\GG)-g(\FFF)|  & \leq |g(\GG)-g(\G(q_0,\GG))|+ 
|g(\G(q_0,\WW_\GG))-g(\G(q_0,\GG))| \nonumber \\
&  \quad + |g(\G(q_0,\WW_\GG))-g(\G(q_0,\WW_\FFF))| + |g(\G(q_0,\FFF))-g(\G(q_0,\WW_\FFF))| \nonumber \\ & \quad +|g(\FFF)-g(\G(q_0,\FFF))|, 
\end{align*}
 and the testability property of $g$ together with (\ref{eq12}) gives the desired result.  
\end{proof}

\section{Parameters depending on densities of linear hypergraphs}\label{sec.linear}

We present a special case of the above notion of ND-testability that preserves several useful properties of the graph case, $r=2$. Restricting our attention to this sub-class we are able to essentially remove the dependence on $r$ in the bound given by \Cref{mainthm} on the sample complexity. 

A linear $r$-graph is an $r$-graph that satisfies that any distinct pair of its edges intersect at most in one vertex. A linear $k$-colored $r$-graph has absent edges, if we disregard the colors of the edges present, then they form a linear $r$-graph. We call an $r$-graph parameter \emph{linearly ND-testable} if it is ND-testable and its witness parameter does only depend on the $t^*$-densities of linear hypergraphs.

In this section we depart from the graphon notion and use instead objects called \emph{naive $r$-graphons} and \emph{naive $r$-kernels}. These differ from true graphons and kernels in their domain that is the $r$-dimensional unit cube and whose coordinates correspond to nodes of $r$-edges instead of any proper subset of the set of nodes of an $r$-edge. They can be transformed into true graphons by adding dimension to the domain in a way that the values taken do not depend on the entries corresponding to the new dimensions. This way we can think of naive graphons as a special subclass of graphons,  sampling is defined analogously to the general case. Note that for $r=2$ the naive notion does not introduce any restriction as all proper subsets of a $2$-element set are singletons. We require the notion of ground state energies of $r$-graphs, naive $r$-graphons, and kernels form \cite{BCL2}, see also \cite{AVKK2}.

Let $s\geq 1$ $J$ be an $r$-array of size $s$, and $G$ be an arbitrary $r$-graph. Define the ground state energy (GSE) (see \cite{BCL2}) of the $r$-graph $G$ with respect to the $r$-array $J$ by
\begin{align}\label{defgse}
\hat \Gamma(G,J)=\max_{\QQ} \sum_{i_1, \dots, i_r=1}^s J(i_1, \dots, i_r) \int_{[0,1]^r} \prod_{j=1}^r \I_{Q_{i_j}}(x_j) W_G(x_1, \dots, x_r) \du x,
\end{align}
where the maximum runs over all partitions $\QQ$ of $[0,1]$ into $s$ parts.
Analogously, define the GSE of a naive $r$-kernel $U$ with respect to $J$ by
\begin{equation*} 
\Gamma (U,J)=\max_{f} \sum_{i_1, \dots, i_r=1}^s J(i_1, \dots, i_r) \int_{[0,1]^r} \prod_{j=1}^r f_{i_j}(x_j) U(x_1, \dots, x_r) \du x, 
\end{equation*}	
where the maximum runs over all fractional partitions $f$ of $[0,1]$ into $s$ parts.

The next result was first proved in \cite{AVKK2}, subsequently refined in \cite{KM1}.

\begin{theorem}\label{maincor1}
	Let $r\geq 1$, $s \geq 1$, and $\delta >0$. Then for any $r$-kernel $U$, real $r$-array $J$, and $q \geq \Theta^4  \log(\Theta)$ with $\Theta=\frac{2^{r+10}s^r r}{\delta}$  we have that
	\begin{align}\label{eqq}
	\PPP(|\Gamma(U,J)-\hat \Gamma(\G (q,U),J)|>\delta \|U\|_\infty)< 2 \exp\left( -\frac{\delta^2q}{8r^2}\right).
	\end{align}
\end{theorem}

We require the version of the norms and distances given in \Cref{sec.def} for the naive setting.
\begin{definition} \label{ch6:def.cutstarnorm}
	
%	Let $r \geq 1$ and $A$ be a real $r$-array of size $n$. Then the cut-$*$-norm of $A$ is
%	\begin{align*}
%	\|A\|^*_{\square,r}=\frac{1}{n^r} \max_{\substack{S_i \subset [n] \\ i \in [r]}}
%	|A(S_1, \dots, S_r)|,
%	\end{align*}
%	where $A(S_1, \dots, S_r)=\sum_{i_1 \in S_1,\dots, i_r \in S_r} A(i_1, \dots, i_r)$.
%	
%	If $\P=(P_i)_{i=1}^t$ is a partition of $[n]$ , then the cut-$(*,\P)$-norm of $A$ is
%	\begin{align*}
%	\|A\|^{*}_{\square,r,\P}=\frac{1}{n^r} \max_{\substack{S_i \subset [n], i \in [r]}}
%	\sum_{j_1, \dots, j_r=1}^t |A(S_1\cap P_{j_1} , \dots, S_r \cap P_{j_r})|.
%	\end{align*}
	
	The cut-*-norm of a naive $r$-kernel $W$ is
	\begin{align*}
	\|W\|^{*}_{\square,r}=\sup_{S_{i} \subset [0,1], i \in [r]}
	\left|\int_{ S_1 \times \dots \times S_r } W(x) \du \lambda(x) \right|,
	\end{align*}
	where the supremum is taken over measurable sets $S_i \subset [0,1]$ for each $i \in [r]$. 
	Furthermore, for a  partition $\P=(P_i)_{i=1}^t$ of $[0,1]$ the cut-$(*,\P)$-norm of a naive $r$-kernel is defined by
	\begin{align*}
	\|W\|^{*}_{\square,r,\P}=\sup_{S_{i} \subset [0,1], i \in [r]}
	\sum_{j_1, \dots, j_r=1}^t \left|\int_{ (S_1\cap P_{j_1})  \times \dots \times (S_r \cap P_{j_r}) } W(x) \du \lambda(x) \right|.
	\end{align*}
	The cut-$(*,\P)$-distance $d^{*}_{\square,r,\P}$ of graphs and graphons is defined analogously to \Cref{defdist} exchanging the cut-$\P$-norm for the cut-$(*,\P)$-norm.  
\end{definition}
%\begin{definition}
%	Let $W$ be a graphon and $\PPP=(P_1, \dots, P_t)$ a partition of $[0,1]$. Then the cut-$\PPP$-norm of $W$ is 
%	\begin{align}
%	\|W\|_{\square\PPP}= \max_{S_i, T_i \subset P_i} \sum_{i,j=1}^t \left|\int_{S_i \times T_j} W(x,y) \du x \du y \right|.
%	\end{align}
%	For two graphons $U$ and $W$ let $d_{W,\PPP}(U)$ denote the cut-$\PPP$-entropy of $U$ with respect to $W$ that is defined by
%	\begin{align}
%	d_{W,\PPP}(U)= \inf_{\phi} \|U^\phi-W\|_{\square\PPP},
%	\end{align}
%	where the infimum runs over all measure preserving maps from $[0,1]$ to $[0,1]$.
%	
%	For $n\geq 1$, a partition $\PPP$ of $[n]$ and a directed weighted graph $H$ the cut-$\PPP$-norm of $H$ on $[n]$ is defined as   
%	\begin{align}
%	\|H\|_{\square\PPP}=\|W_H\|_{\square\PPP'},
%	\end{align}
%	where $\PPP'$ is the partition of $[0,1]$ induced by $\PPP$ and the map $j \mapsto [\frac{j-1}{n},\frac{j}{n})$.
%\end{definition}

The definition for the $k$-colored version is analogous.

We require the following auxiliary lemmas that are analogous to \Cref{ch6:sregkhyper}, \Cref{ch6:hypercountcor}, and \Cref{ch6:hypercoloring}, respectively (with analogous proofs).

\begin{lemma}\label{sregkhyper_lin}
	%Let $n \geq 1$ fixed, and let each partition in the statement be such that it is refined by the canonical partition of $[0,1]$ into $n$ parts, and each function  be such that it is constant on the product sets of classes of the canonical partition.
	
	For every $r \geq 1$, $\varepsilon>0$, $t \geq 1$, $k \geq 1$ and $k$-colored  $r$-graphon $\WW$ there exists a partition $\P=(P_1, \dots, P_m)$ of $[0,1]$ into $ m \leq (2t)^{(rk+1)^{4k^2/\varepsilon^2}}=t_\reg(r,k,\varepsilon, t)$ parts and a naive $(r,1)$-step function $\VV \in \WWW^{r,k}$ with steps from $\P$, such that for any partition $\QQ$ of $[0,1]$ into at most $m t$ classes we have
	\begin{align*}
	d^{*}_{\square,r,\QQ}(\WW,\VV) \leq \varepsilon. 
	\end{align*} 
	%Furthermore it holds that
	%\begin{align}
	%d_{\square,r,\P}(\WW,\WW_\P) \leq \varepsilon. 
	%\end{align}
	%If $\WW=\WW_\GG$ for some $k$-colored $\GG$ with $|V(\GG)|=n$, then one can require that $\P$ is an $n$-partition of $[0,1]^{\hhh([r-1])}$. 
	%If we want the parts to have almost equal measure, then the upper bound on the number of classes is modified into $2^{(4k^2+2)^{64/\varepsilon^2}}$.
\end{lemma}

\begin{lemma}
	\label{lincountlemma}
	Let $U$ and $W$ be $k$-colored $r$-kernels with $\|U\|_\infty, \|W\|_\infty \leq 1$. Then for every linear $k$-colored $r$-graph $F$ we have
	\begin{align*}
	|t^*(F,W)-t^*(F,U)| \leq {q \choose r} 	d^{*}_{\square,r}(U,W). 
	\end{align*}
\end{lemma}

\begin{lemma}\label{lincollemma}
	Let $k\geq 1$, $\varepsilon>0$, $U$ be a step function with steps $\P=(P_1, \dots, P_t)$ and $V$ be a $r$-graphon with $d^{*}_{\square,r,\P}(U,V) \leq \varepsilon.$ For any $k$-colored $r$-graphon $\UU=(U^{1},\dots,U^{k})$ that is a step function with steps from $\P$ and a $k$-coloring of $U$  there exists a $k$-coloring $\VV=(V^{1},\dots,V^{k})$ of $V$ so that $d^{*}_{\square,r}(\UU,\VV)= \sum_{\alpha=1}^k \|U^{(\alpha)}-V^{(\alpha)}\|_\square \leq k\varepsilon.$
	
%	If $V=W_G$ for a simple graph $G$ on $n\geq 16/\varepsilon^2$ nodes and $\P$ is an $n$-partition of $[0,1]$ then there is a $(k,m)$-coloring $\GG$ of $G$ that satisfies the above conditions and $\|\UU-\WW_\GG\|_\square \leq 2k^2\varepsilon.$
	
\end{lemma}

Next we state and prove the main contribution of this section.

\begin{theorem}\label{linpropthm}
	Let $f$ be a linearly ND-testable $r$-graph parameter with witness parameter $g$ of $k$-colored $r$-graphs, and let the corresponding sample complexity be $q_g$. Then $f$ is testable with sample complexity $q_f$, and  there exists a constant $c>0$ only depending on $k$ and $r$ but not on $f$ or $g$ such that for any $\varepsilon>0$ we have \begin{align}
		q_f(\varepsilon) \leq \exp^{(3)}(cq^2_g(\varepsilon/2)).\end{align}
\end{theorem}

\begin{proof}
	The proof is almost identical to case of graphs in \citet{KM2}, we will sketch it in the framework of \Cref{mainlemma}, from there the statement follows a similar way as the proof of \Cref{mainthm}. The main distinction between the general setting and the current linear setting is that we do not require for each coloring $\VV$ of $V$ to have a corresponding coloring $\UU$ of $U$ such that their $q_0$-sampled distribution are close in the total variation distance, here it is enough to impose that they are close in $d^{*}_{\square,r}$. This relaxed condition implies that the conditional $q_0$-sampled distributions are close, where the condition comprises the densities of linear sub-hypergraphs. The different norm employed in the measurements of the proximity allows us to remove the inductive part that is contained in the general proof in \Cref{mainlemma}. 
	
	 Let $f$  and $g$ be such as in the statement of the theorem, and let $G$ be an arbitrary $r$-graph and $W_G$ a $3$-colored naive $r$-graphon that represents it (the colors correspond to edges, non-edges, and diagonal entries respectively). Let $q \geq  \exp^{(3)}(cq^2_g(\varepsilon/2))$ for some $c>0$ that is chosen large enough, and let $F$ denote the random $r$-graph $\G(q,G)$, and let $W_F$ be its $3$-colored representative graphon. It is easy to see as in the general case that $f(F)\geq f(G)-\varepsilon/4$ with probability at least $1-\varepsilon/4$, in fact this is even true with much smaller $q$. 
	 
	 We will show first that with probability at least $1-\varepsilon/4$ there  exist for every $k$-coloring $\VV=(V^{\alpha,\beta})_{\alpha \in [3], \beta \in [k]}$ of $W_F$ a $k$-coloring $\UU=(U^{\alpha,\beta})_{\alpha \in [3], \beta \in [k]}$  of $W_G$ such that 	$d^{*}_{\square,r,\QQ}(\UU,\VV) \leq \Delta$, where $\Delta=\exp(-c'q^2_g(\varepsilon/2))$. 
	 Let $W_1$ be a naive $r$-graphon that satisfies 	
	 \begin{align}\label{eq1}
	 \sup_{t_\QQ \leq t_\P t_2}d^{*}_{\square,r,\QQ}(W_G,W_1) \leq \Delta/8k, 
	 \end{align}
	 by \Cref{sregkhyper_lin} there exists such a naive $(r,1)$-step function with at most $t_1=t_\reg(r,2,\Delta/8k, t_2)$ steps that are denoted by $\P$, where $t_2=t_\reg(r,3k,\Delta/8k, 1)$. 
	 Further, let $W'_2$ be the naive $(r,1)$-step function associated with $\G(q,W_1)$ with its steps forming  the partition $\P''$. There exists a measure-preserving permutation $\phi$ of $[0,1]$ such that  $W_2$ given by $W_2(x_1, \dots, x_r)=W'_2(\phi(x_1), \dots, \phi(x_r))$ is another valid representation of $\G(q,W_1)$ with steps $\P'$, and having the additional property that the measure of the set where $W_1$ and $W_2$ differ is at most  
	 $r \sum_i |\lambda(P_i)-\lambda(P''_i)|$. In particular by the choice of $q$ it is true that $\|W_1-W_2\|_1\leq \Delta/8k$ with probability at least $1-\varepsilon/8$.
	 
	 Further, the bound in (\ref{eq1}) can be rewritten as a GSE problem in the sense of (\ref{defgse}), applying \Cref{maincor1} leads to the assertion that 
	   \begin{align}\label{eq2}
	   \sup_{t_\QQ \leq t_\P t_2}d^{*}_{\square,r,\QQ}(W_F,W_2) \leq \Delta/4k,
	   \end{align}
	   with probability at least $1-\Delta/8k$, which is larger than $1-\varepsilon/8$.
	 
	 We condition on the aforementioned two events, they occur jointly with probability at least $1-\varepsilon/4$. Now let $\VV$ be an arbitrary $k$-coloring of $W_F$, it follows that there exists a $3k$-colored naive $(r,1)$-step function $\ZZ=(Z^{\alpha,\beta})_{\alpha \in [3], \beta \in [k]}$ with steps forming $\RRR$ such that 
	  \begin{align}\label{eq3}
	  \sup_{t_\QQ \leq t_\RRR }d^{*}_{\square,r,\QQ}(\VV,\ZZ) \leq \Delta/4k, 
	  \end{align}
	  and $t_\RRR\leq t_2$.
	  Let the naive $r$-graphon $Z$ denote the $k$-discoloring of $\ZZ$. Then we have
	 \begin{align}\label{eq4}
	 \sup_{t_\QQ \leq t_\RRR }d^{*}_{\square,r,\QQ}(W_F,Z) \leq \Delta/8k, 
	 \end{align}
	 and together with (\ref{eq2}) it follows that
	 \begin{align}\label{eq5}
	 \sup_{t_\QQ \leq t_\RRR }d^{*}_{\square,r,\QQ}(W_2,Z) \leq \Delta/4k.
	 \end{align}
	  
	  An application of \Cref{lincollemma} together with the bound in (\ref{eq5}) ensures the existence of a $k$-coloring $\WW_2$ of $W_2$ that is a naive $(r,1)$-step function with the steps comprising $\Ss$ that is the coarsest common refinement of $\P'$ and $\RRR$, and that satisfies  
	  \begin{align}\label{eq6}
	  d^{*}_{\square,r}(\WW_2,\ZZ) \leq \Delta.
	  \end{align}
	  Now we construct a $k$-coloring of $W_1$ by simply copying $\WW_2$ on the set on $[\cup_i (P_i \cap P'_i)]^r$, and defining it in arbitrary way on the rest of $[0,1]^r$ paying attention to keep it a $k$-coloring of $W_1$ and not increase the number of steps above $t_\RRR$.
	For the $\WW_1$ obtained this way we have 
      \begin{align}\label{eq7}
     \sum_{\alpha, \beta}\|W_1^{\alpha, \beta}-W_2^{\alpha, \beta} \|_1=d_1(\WW_1,\WW_2) \leq \Delta/4.
     \end{align}	
     Employing again \Cref{lincollemma} with (\ref{eq1}) we obtain a $k$-coloring $\UU$ of $W_G$ that satisfies 
      \begin{align*}
      d^{*}_{\square,r}(\UU,\WW_1) \leq \Delta,
       \end{align*}
       hence
       \begin{align*}
       d^{*}_{\square,r}(\UU,\VV) \leq 4\Delta.
       \end{align*}
       With a further randomization we can form a proper $k$-coloring $\GG$ of $G$ that satisfies
       \begin{align*}
              d^{*}_{\square,r}(\WW_\GG,\VV) \leq 5\Delta.
              \end{align*}
        Finally, we use that 
       \begin{align*}
       |g(\FFF)-g(\GG)|  & \leq |g(\FFF) - g(\G(q_g(\varepsilon/4),\FFF))|+|g(\GG)-g(\G(q_g(\varepsilon/4),\GG))| \leq \varepsilon/2,
       \end{align*}
       whenever there exists a coupling of the random $2k$-colored $r$-graphs $\G(q_g(\varepsilon/4),\GG)$ and $\G(q_g(\varepsilon/4),\FFF)$ appearing in the above formula such that their densities of linear subgraphs  are equal with probability larger than $\varepsilon/2$. Such a coupling exists by \Cref{lincountlemma} and standard probabilistic assumptions, thus we have 
       $f(G)\geq f(F)-\varepsilon/2$ with probability at least $1-\varepsilon/4$, that concludes the proof.
       
\end{proof}

\section{Applications}\label{sec.appl}

The characterization of testability of properties of $r$-uniform hypergraphs for $r\geq 3$ is a well-studied area, for instance it has been established by \citet{RSch} that hereditary properties (properties that are preserved under the removal of vertices) are testable generalizing the situation in the graph case. Nevertheless, several analogous question to the graph case have remained open. We present some of these in this section together with the proofs for positive results as an application of \Cref{mainthm}.

\subsection{Energies and partition problems}

We define a family of parameters of $r$-uniform hypergraphs that is a generalization of the ground state energies (GSE) of \citet{BCL2} in the case of graphs (see also \Cref{sec.linear}), for connections to statistical physics, in particular to the Ising and the Curie-Weiss model, see \cite{BCL2}. This notion encompasses several important graph optimization problems, such as the maximal cut density and multiway cut densities for graphs, therefore its testability is central to several applications.  

\begin{definition}\label{ch6:def.gse1}
	For an $r$-graph $H \subset {[n] \choose r}$, a real $r$-array $J$ of size $q$, and a symmetric partition $\P=(P^1, \dots, P^t)$ of ${[n] \choose r-1}$ we define the energy
	\begin{equation*}
	\EEE_{\P,r-1} (H,J)=\frac{1}{n^r} \sum_{i_1, \dots, i_r =1}^t J(i_1, \dots, i_r) e_{H} (r; P_{i_1}, \dots, P_{i_r}),
	\end{equation*}
	where $e_{H}(r; S_1, \dots, S_r)= | \{(u_1, \dots, u_r) \in [n]^r | \{u_1, \dots, u_r\}\in H \textrm{ and } \{u_1, \dots, u_{j-1}, u_{j+1}, \dots, u_r\} \in S_j  \textrm{ for all } j=1, \dots, r\} |$. 
	
	Let $\HHH=(H^\alpha)_{\alpha \in [k]}$ be a $k$-colored $r$-uniform hypergraph on the vertex set $[n]$ and $\JJ=(J^{\alpha})_{\alpha \in [k]}$ be a tuple of real $r$-arrays of size $t$ with $\| \JJ \|_\infty \leq 1$. Then the energy for a partition $\P$ as above is 
	\begin{equation*}
	\EEE_{\P,r-1} (\HHH,\JJ)= \sum_{\alpha \in [k]} \EEE_{\P,r-1} (H^\alpha,J^\alpha).
	\end{equation*}
	
	The maximum of the energy over all partitions $\P$ of ${[n] \choose r-1}$ is called the generalized ground state energy (GGSE) of $\HHH$ with respect to $\JJ$, and is denoted by 
	\begin{equation*}
	\EEE_{r-1} (\HHH,\JJ)=\max_\P \EEE_{\P,r-1} (\HHH,\JJ).
	\end{equation*}
\end{definition}

A rather straightforward application of \Cref{mainthm} gives us the testability of any GGSE. 
\begin{corollary}
For any $r,q\geq 1$ and real $r$-array $J$ of size $t$ the generalized ground state energy  	$\EEE_{r-1} (. ,J)$ is a testable $r$-graph parameter.
\end{corollary}

We note that this result was proved previously in \cite{KM3}, Theorem 3.15., the proof there used ultralimits and was therefore non-effective. The present corollary does not rely on such tools, we could provide an explicit upper bound on the sample complexity, and in this sense the result is new.

The above problem of testing of the GGSE is a special case of the question regarding testability of general partition problems. These properties were first dealt with systematically in the graph case in \cite{GGR}, where the authors showed their testability. They are also the most prominent family of non-trivial properties from the testing perspective in the dense model that are testable with polynomial sample complexity known to date.

We sketch briefly the problem. Consider a vector of $k$ positive reals adding up to $1$ and a symmetric matrix of size $k$ with entries from $[0,1]$ together forming a so-called density tensor. The partition property associated to this tensor is satisfied by a graph whenever there exists a partition of its vertex set so that the densities of the class sizes equal the quantities given by the vector and the edge densities between the parts coincide with the corresponding entries of the matrix. A property associated to a family of tensors is satisfied whenever there exists a member of the family that is satisfied following the above description. For example, we can throw away from a density tensor the condition on the class sizes, or we can require that the edge densities between the classes lie in a certain interval to obtain another, relaxed partition problem. 

A test for the maximal cut density can be obtained from a collection of partition problems into two classes only constraining the edge density between the two distinct parts for each integer multiple of $\varepsilon$ in $[0,1]$.

Researched aimed at partition problems for hypergraphs was initiated by \citet{FMS} defining a framework that slightly extended the notions of \cite{GGR}. In their setup the problem is formulated again as a question of existence of a vertex partition of a hypergraph with prescribed sizes that satisfies that the $r$-partite sub-hypergraphs spanned  by each $r$-tuple of classes contain a certain number of edges. The additional feature of the approach is that it can also handle tuples of uniform hypergraphs (perhaps of different order) that share a common vertex set that is the subject of the partitioning, the partition problem defined again by density tensors comprises of constraints on edge densities between classes for each of the component hypergraphs. In \cite{FMS} it is shown that such properties are testable with polynomial sample complexity.

A further generalization has been investigated by \citet{Roz} dealing the first time with constraints imposed on partitions of pairs, triplets, and so on of the vertices on one hand, and the edge densities filtered by these partitions on the other. However the edge density constraints in \cite{Roz} are not partitioning the edge set as in the previous approaches, rather layers of partitions corresponding to partitions of $[r]$ for $r$-graphs are considered. Let us illustrate the framework for $3$-graphs with the partitioning understood as coloring. In \cite{FMS} the number of edges whose vertices have certain colors are constrained, in  \cite{Roz} also the number of edges can be constrained that fulfill the condition that a pair of vertices (as a tuple) has a certain color and the  third vertex (as a singleton) has also some other color. However, in \cite{Roz} only colorings disjoint subsets of the $r$-edges are allowed to yield a constraint, for instance it is not possible to have a condition on the number of pair-monochromatic edges, that is, $3$-edges whose three underlaying pairs have the same color. The positive result obtained in \cite{Roz} is also somewhat weaker than testability, the term pseudo-testability is introduced in order to formalize the conclusion.

Our approach allows for more general constraints on edge densities.

\begin{definition}
	Let $\Phi$ denote the set of all maps $\phi$ that are assigning to each element of the set of proper subsets of $[r]$, $\hhh([r],r-1)$, a color $[k]$. We define a density tensor by 
	$\psi=\langle \langle \rho^s_i \rangle_{s \in \hhh([k],r-1)},\langle \mu_{\phi} \rangle_{\phi \in \Phi} \rangle$, where each component is in $[0,1]$.
	
	Let $H$ be an $r$-graph with vertex set $V=V(H)$ of cardinality $n$ and for each $1 \leq s \leq r-1$ let  $\P(s)$ be a partition of ${V \choose s}$ into $k$ parts, and let $\P=(\P(s))_{s=1}^{r-1}$. Then the density tensor corresponding to the pair $(H, \P)$ is given by 
	\begin{align*}
	 \rho^s_i(H, \P)= \frac{|P_i(s)|}{n^s} \quad \textrm{for all} \quad s \in \hhh([k],r-1),
	 \end{align*}
	 and 
	 \begin{align*}
	 \mu_{\phi}(H, \P)=\frac{|\{e \in [n]^r| {\overline e} \in H \textrm{ and } \quad \overline{p_A(e)} \in P_{\phi(A)}(|A|) \quad  \textrm{for all } A \in  \hhh([r],r-1)  \}|}{n^r} \quad \textrm{for all} \quad \phi \in \Phi,
	 \end{align*}
	 where ${\overline v}$ is the set that consists of the components of the vector $v$.
	 We say that $H$ satisfies a density tensor $\psi$ if there exists a collection of partitions $\P$ of its vertex tuples as above so that the tensor yielded by the pair $(H,\P)$ is equal to $\psi$.
\end{definition}
	
We remark that the above partition property is non-hereditary. An application of our main result yields the following corollary.
\begin{corollary}\label{partpcor}
For any $r\geq 1$ and a density tensor $\psi=\langle \langle \rho^s_i \rangle_{s \in \hhh([k],r-1)},\langle \mu_{\phi} \rangle_{\phi \in \Phi} \rangle$, the partition property given by the tensor is testable.
\end{corollary}

\subsection{Logical formulas}

The characterization of testability in terms of logical formulas was initiated by \citet{AFKSz} who showed that properties expressible by certain first order formulas are testable, while there exists some first order formulas that generate non-testable properties. The result can be formulated as follows. 

\begin{theorem}\cite{AFKSz}\label{afksz}
Let $l,k\geq 1$ and $\phi$ be a quantifier-free first order formula of arity $l+k$ containing only adjacency and equality.
The graph property given by the truth assignments of the formula $\exists u_1, \dots, u_l \forall v_1, \dots, v_k  \phi( u_1, \dots, u_l, v_1, \dots, v_k)$ with the variables being vertices is testable.
\end{theorem}

Without going into further details at the moment we mention that any $\exists\forall$ property of graphs is indistinguishable by  a tester from  the existence of a node-coloring that is proper in the sense that the colored graph does not contain subgraphs of a certain set of forbidden node-colored graphs.

Our focus is directed at the positive results of \cite{AFKSz}, those were generalized into two directions.First, by \citet{JZ} to relational structures in the sense that $\phi$ can contain several $r$-ary relations with even $r\geq 3$ whereas the $\exists\forall$  prefix remains the same concerning vertices. Secondly, by \citet{LV} to a restricted class of second order formulas, where existential quantifiers for $2$-ary relationships are added ahead of the above formula in \Cref{afksz} so that they can be included in $\phi$, see Corollary 4.1 in \cite{LV}. %The non-effectiveness of the latter proof was subsequently fixed by the approach of \cite{GS}.
Our framework allows for extending these results even further.

\begin{corollary}
	Let $r_1, \dots, r_m, l,k \geq 1$ be arbitrary, and let $r=\max r_i$. For any $r$-graph property that is expressible by the truth assignments of the second order formula
	\begin{align} \label{eq11}
	\exists L_1, \dots, L_m \exists u_1, \dots, u_l \forall v_1, \dots, v_k \phi(L_1, \dots, L_m, u_1, \dots, u_l, v_1, \dots, v_k),
	\end{align} where $L_i$ are symmetric $r_i$-ary predicate symbols and  $u_1, \dots, u_l, v_1, \dots, v_k$ are nodes, and $\phi$ is a quantifier-free first order expression containing adjacency, equality, and the symmetric $r_i$-ary predicates $L_i$ for each $i \in [m]$	is testable.
\end{corollary}
\begin{proof}[Proof (Sketch)]
We first note that any collection of the relations $L_1, \dots, L_m$ can be encoded into one edge-colored $r$-uniform hypergraph with at most $2^{rm}$ colors with an additional compatibility requirement. An edge color for $e \in {[n] \choose r}$ consists of a $2^r-1$-tuple corresponding to non-empty subsets of $[r]$, where the entry corresponding to $S \subset [r]$ is determined by the evaluation of $p_S(e)$ in the relations $L_i$ that have arity  $|S|$. We can reconstruct the predicates from a coloring whenever the color of any pair of edges $e$ and $e'$ is such that their entries corresponding to the power set of $e \cap e'$ coincide, for $r=2$ this means some combinations of colors (determined by a partition of the colors) for incident edges are forbidden. This compatibility criteria for $2^{rm}$-colored $r$-graphs is known to be a testable property, from here on this will be seen as a default condition.

For a fixed tuple $L_1, \dots, L_m$ of relations of arity at most $r$ the property corresponding to the first order expression $\forall v_1, \dots, v_k \phi(L_1, \dots, L_m, v_1, \dots, v_k)$ is equivalent to the property of $2^{rm}$-colored $r$-graphs that is defined by forbidding certain subgraphs of size at most $k$. This is testable by the following theorem of \citet{AT} that generalizes the result of \citet{RSch}.

\begin{theorem}\label{herethm}\cite{AT}
For any $r,k\geq 1$, every hereditary property of $k$-colored $r$-graphs is testable.
\end{theorem}

We sketch now that the properties corresponding to the more general formula (\ref{eq11}) in the statement of the corollary are indistinguishable from the existence of a further node-coloring on top of the edge-colored graphs such that no subgraph appears from a certain set of forbidden subgraphs. We follow the argument of \cite{AFKSz} (see also \cite{JZ}, and \cite {LV}). Two properties are said to be indistinguishable in this sense whenever for every $\epsilon>0$ there exists an $n_0=n_0(\varepsilon)$ such that any graph on $n\geq n_0$ vertices that has one property can be modified by at most $\epsilon n^r$ edge additions or removals to obtain a graph that has the other property, and vice versa. The testability behavior of the two properties is identical. Consider  $L_1, \dots, L_m$ as fixed, then the property of $2^{rm}$-colored $r$-graphs corresponding to $\exists u_1, \dots, u_l \forall v_1, \dots, v_k \phi(L_1, \dots, L_m, u_1, \dots, u_l, v_1, \dots, v_k)$ is indistinguishable to from the existence of the following proper coloring. Every node gets either color $(0,0)$ or $(a,b)$, where $a$ represents an $2^{rm}$-colored $r$-graph on $l$ nodes, and $b$ represents  an $l$-tuple of $2^{rm}$-colored edges. A coloring is proper if there are at most $l$ nodes colored by $(0,0)$, further for any other color appearing the first component $a$ is identical.
Now a colored subgraph of size $k$ is forbidden if considering the edge-colored graph on $V=\{v_1, \dots ,v_k\}$  (without node colors)  supplemented by a graph on $\{u_1, \dots, u_l\}$ together with their connection to $V$ given by the node colors on $V$ the evaluation of the formula $\phi(L_1, \dots, L_m, u_1, \dots, u_l, v_1, \dots, v_k)$ is false.

It is not hard to see that for this coloring property  \Cref{herethm} applies since it is hereditary, therefore it is testable. Now if we let $L_1, \dots, L_m$ to be arbitrary and apply \Cref{mainthm}, then we obtain the testability of the property given by (\ref{eq11}) in the statement of the corollary.
\end{proof}

\subsection{Estimation of the distance to properties}
We can also express the property of being close to given property in the nondeterministic framework, and can show the testability here. This problem was introduced first for graphs by \citet{FN}, in this paper the authors show the equivalence of testability and estimability of the distance of a property, in \cite{LV} one direction of this was reproved for graphs. To our knowledge the generalization for $r$-graphs has not been considered yet. Recall that $d_1$ is the edit distance.

\begin{corollary}
For any $r\geq 1$, testable $r$-graph property $\P$ and real $c>0$ the property  $d_1(., \P)< c$ is testable.
\end{corollary}
\begin{proof}
The proof is identical to the one given in \cite{LV}, for any $r\geq 1$, testable $r$-graph property $\P$ and real $c>0$ a testable property of $4$-colored $r$-graphs that witnesses the property of $d_1(., \P)< c$. Let $G$ be an arbitrary $r$-graph, then we consider the $2$-colorings of $G$ where $(1,1)$ and $(1,2)$ color the edges of $G$, and $(2,1)$ and $(2,2)$ the non-edges. The $4$-colored witness property $\QQ$ is then that the edges with the colors $(1,1)$ and $(2,1)$ together form a member of $\P$, and additionally there are at most $cn^r$ edges colored by $(1,2)$ or $(2,1)$. The property $\QQ$ is trivially testable, therefore \Cref{mainthm} implies the statement.
\end{proof}
\section{Further research} \label{sec.fur}

The general upper bound given in \Cref{mainthm} is dependent on the order $r$, it would be interesting to see if it is possible to remove this dependence in a similar way as it was shown in the special case of linearly ND-testable parameters. A more ambitious goal would be to transform effective bounds into efficient if possible. We mean by this the verification that the sample complexity of the original parameter or property is of the same magnitude (up to polynomial dependence) as the sample complexity of the witness parameter. Currently no non-trivial lower bound on the sample complexity in our framework is known, in the original dense property testing setting there are some properties that admit no tester that only makes a polynomial number of queries, such as triangle-freeness and other properties defined by forbidden families of subgraphs or induced subgraphs.

The partition problems described in \Cref{sec.appl} had lead to further applications in the graph case, this development was presented in \cite{FMS}. As mentioned, the framework of \cite{FMS} also dealt with tuples of hypergraphs extending the result of \cite{GGR}, this enabled the analysis of the number of $4$-cycles appearing in the bipartite graphs induced by the pairs of the partition classes instead of only observing the edge density by means of adding an auxiliary $4$-graph to the simple graph. An alternative characterization of the notion of a regular bipartite graph says that a pair of classes is regular if and only if the number of $4$-cycles spanned by them is minimal, with other words their density is approximately the fourth power of the edge density. Using this together with the result  regarding the testability of partition problems of \cite{FMS} the authors there were able to show that satisfying a certain regularity instance is also testable. This achievement in turn imply an algorithmic version of the Regularity Lemma. In this manner, \Cref{partpcor} might be of further use for testing regular partitions of $r$-uniform hypergraphs by utilizing concepts that emerged during the course of research towards an algorithmic version of the Hypergraph Regurality Lemma (see for example \citet{HNR}) in a similar way to the approach in \cite{FMS}.

On a further thought, one may depart from the setting of dense $r$-graphs in favor of other classes of combinatorial objects in order  to define and study their ND-testability. Such are for example semi-algebraic hypergraphs that admit a regularity lemma that produces a polynomial number of classes as a function of the multiplicative inverse of the proximity parameter, thus they are good candidates for an improvement on the sample complexity.

Finally we mention a possible direction for further study towards the characterization of locally repairable properties, see \cite{AT}, that appears to be promising. This characteristic is stronger than testability in that respect that in this setup there should exist a local edge modifying algorithm applied to graphs that are close to the property that observes only some piece of bounded size of the graph and its connection to single vertex pairs and decide upon adjacency depending only on this information. The output of this algorithm should be a graph that is close to the input and actually satisfies the property. We may define nondeterministically locally repairable properties in a straight-forward way analogous to ND-testing by requiring a certain locally repairable property of edge-colored graphs that reduces to the original property after the discoloring procedure. It has been established in \cite{AT} that hereditary graph properties are locally repairable, but there are examples of hereditary properties of directed graphs and $3$-graphs that are testable, but not locally repairable. It would be compelling to investigate analogous problems concerning nondeterministically locally repairable properties.

\bibliographystyle{notplainnat}
\bibliography{thesis_refs}

\end{document}